\documentclass[a4paper, 11pt]{scrartcl}

\usepackage{amsmath, amssymb, fancyhdr , graphicx, setspace, lastpage, bbm, amsthm} 
\usepackage{graphics}
\usepackage{epsfig}
\usepackage{epstopdf}
\usepackage{longtable}
\usepackage{commath} 
\usepackage[top=4.5cm, bottom=3cm, left=3cm, right=3cm]{geometry}
\usepackage[T1]{fontenc}
\usepackage{lmodern}
\usepackage[scientific-notation=true]{siunitx}
\usepackage{enumerate, stmaryrd, wasysym}
\usepackage{booktabs}
\usepackage{courier}
\usepackage{xfrac} 
\usepackage{hyperref}
\usepackage{tikz}                                  
\usetikzlibrary{matrix,arrows, mindmap,decorations.pathmorphing}

\usepackage{subfig}
\usepackage{framed}
\usepackage[]{natbib}
\usepackage[title,titletoc,toc]{appendix}

\usepackage[final]{changes}
\usepackage{floatrow}
\newfloatcommand{capbtabbox}{table}[][\FBwidth]
\usepackage{blindtext}
\usepackage[final]{showkeys}

\usepackage{etoolbox}
\newtoggle{submission}
\toggletrue{submission}

\usepackage{setspace}

\newtheorem{theorem}{Theorem}

\newtheorem{remark}{Remark}
\newtheorem{ass}{Assumption}
\newtheorem{prop}{Proposition}

\graphicspath{{images/}{.}}

\newcommand{\E}[2]{\mathbb{E}_{#1}\left[ #2\right]}

\newcommand{\Ef}[2]{\hat{\mathbb{E}}_{#1}\left[ #2\right]} 

\newcommand{\hatL}{L^{l\rightarrow c}}
\newcommand{\1}[1]{\mathbbm 1_{#1}} 
\newcommand{\vet}[1]{\textbf{#1}}
    
\DeclareMathOperator{\rpv}{RPV01}
\DeclareMathOperator{\lgd}{LGD}        
\newcommand{\process}[1]{(#1 _t, t\geq 0) } 
\newcommand{\EUR}[1]{#1 _{\textrm{EUR}} } 
\newcommand{\USD}[1]{#1 _{\textrm{USD}} } 
\newcommand{\Y}[2]{#1 ^{\textrm{{#2}Y}} } 

\author{Damiano Brigo\thanks{Imperial College, London, U.K.
({\texttt{damiano.brigo@imperial.ac.uk}})}
\and\
Nicola Pede\thanks{Imperial College, London, U.K.
({\texttt{n.pede13@imperial.ac.uk}})}
\and\
Andrea Petrelli\thanks{Credit Suisse, London, U.K.
({\texttt{andrea.petrelli@credit-suisse.com}}). }
}

\title{
Multi Currency Credit Default Swaps
}
\subtitle{
Quanto effects and FX devaluation jumps
}

\date{
First posted on SSRN and arXiv on December 2015\\
Second version posted on SSRN on February 2017\\
This version: \today}

\begin{document}
\onehalfspacing
\maketitle

\begin{abstract}
Credit Default Swaps (CDS) on a reference entity may be traded in multiple currencies, in that protection upon default may be offered either in the currency where the entity resides, or in a more liquid and global foreign currency. In this situation currency fluctuations clearly introduce a source of risk on CDS spreads.  For emerging markets, but in some cases even in well developed markets, the risk of dramatic Foreign Exchange (FX) rate devaluation in conjunction with default events is relevant. We address this issue by proposing and implementing a  model that considers the risk of foreign currency devaluation that is synchronous with default of the reference entity. As a fundamental {case} we consider the sovereign CDSs on Italy, quoted both in EUR and USD.

Preliminary results indicate that perceived risks of devaluation can induce a significant basis across domestic and foreign CDS quotes. For the Republic of Italy, a USD CDS spread quote of 440 bps can translate into a EUR quote of 350 bps in the middle of the Euro--debt crisis in the first week of May 2012. More recently, from June 2013, the basis spreads between the EUR quotes and the USD quotes are in the range around 40 bps. 

We explain in detail the sources for such discrepancies. Our modeling approach is based {on the reduced form framework for credit risk, where the default time is modeled in a Cox process setting with explicit diffusion dynamics for default intensity/hazard rate and exponential jump to default. For the FX part,  we include an explicit default--driven jump in the FX dynamics. As our results show}, such a mechanism provides a further and more effective way to model credit / FX dependency than the instantaneous correlation that can be imposed among the driving Brownian motions{ of default intensity and FX rates, as it is not possible to explain the observed basis spreads during the Euro--debt crisis by using the latter mechanism alone}.
\end{abstract}

{\textbf{AMS Classification Codes} }: 60H10, 60J60, 91B70;


\medskip


\medskip


\setcounter{tocdepth}{2}  
\tableofcontents

\newpage


\section{Introduction}


\subsection{Overview of the {M}odelling {P}roblem}
{The need for q}uanto default modeling arises naturally when pricing credit derivatives offering protection  in {multiple currencies}. 

Reasons for entering into {\emph{Credit Default Swaps} (}CDS{)} in different currencies can come from financial, economic{,} or even legislative considerations{:} they range from the composition of the portfolio that has to be {hedged} to the accounting rules in force in the country where the investor is based. 
{I}n case the reference entity is sovereign, economic reasons play a major role {since} for an investor it might be more appealing to buy protection on{, for example, } Republic of Italy's default in USD rather than in EUR{.
Indeed,} in the latter case the currency value itself is strongly related with the reference entity's default.

{Figure \ref{fig:histSpreads} shows the time series of par--spreads for USD--denominated and EUR--denominated CDSs on Republic of Italy from the beginning of 2011 until the end of 2013. 
The time range has been chosen so to include the 2011 Euro--debt crisis.}

\begin{figure}[!tbh]
\centering
\includegraphics[width = \textwidth]{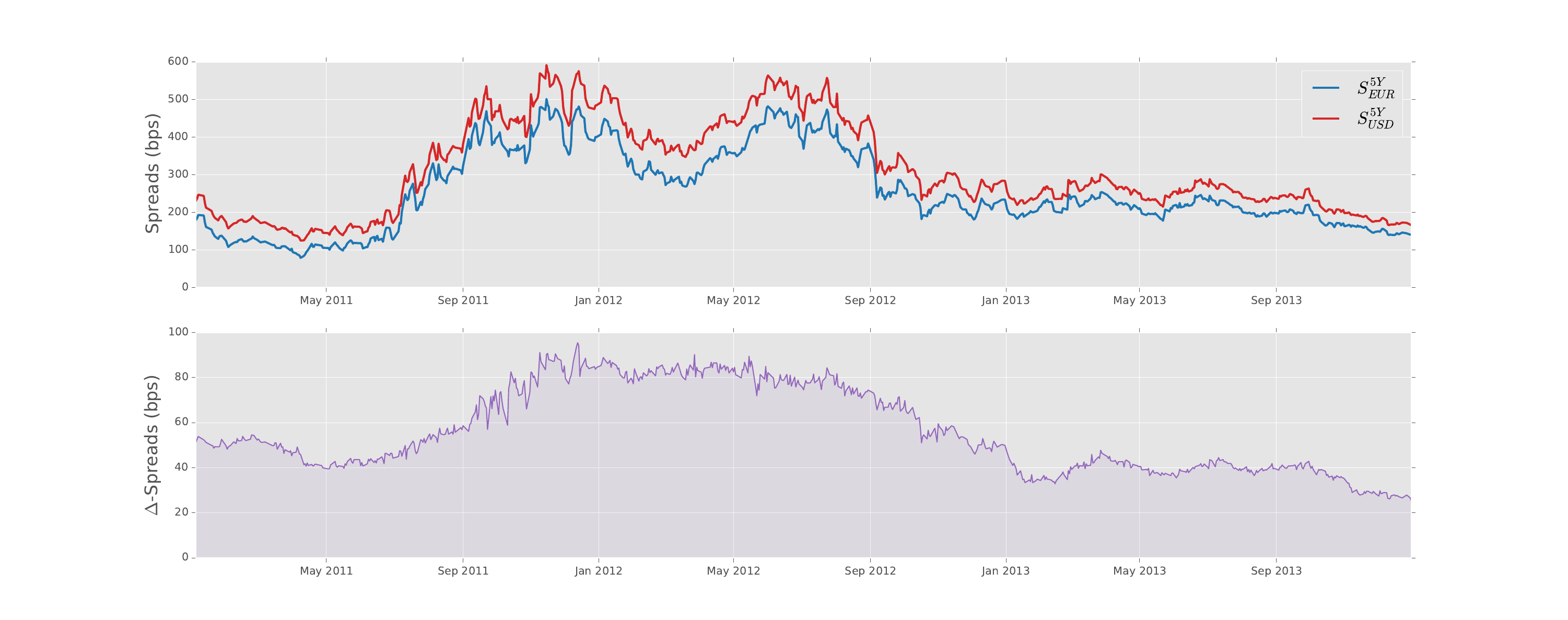}
 \caption{{In the top chart, }5Y par spread time series for USD--denominated CDSs, $\USD{S}^{5Y}$, and EUR--denominated CDSs, $\EUR{S}^{5Y}$, on Italy. The difference between the two par--spreads is {showed in the bottom} chart.}%
\label{fig:histSpreads}
\end{figure}

{
The difference between the par--spreads for USD--denominated and EUR--denominated CDSs is {shown in} the bottom chart. 
In order to build a model which accounts for the default information and  generate the spreads in the two currencies, the joint evolution of the obligors hazard rate and of the FX rate between the two currencies must be modelled.
}

{
In the present paper we  show two ways to model the joint dynamics of credit and FX rate{s}. 
In the first approach  the interaction between the credit and the FX component is entirely explained by an instantaneous correlation between the Brownian motions driving the {stochastic} hazard rate and the FX rate.
In a second{,} more sophisticated modelling approach, a further source of dependence between the two components is introduced in the form of a conditional devaluation {jump} of the FX {r}ate upon default of the reference obligor.
}

{The diffusive approach  emphasizes the limitations of confining  the credit/market interaction to instantaneous correlation between hazard rate and market risk factors. 
As shown {by comparing the model--implied quanto spreads in Figure \ref{fig:rhoSpread} with the observed quanto spreads }{in  Figure \ref{fig:histSpreads}}, instantaneous correlation alone is not able to explain the observed quanto spread. 
This  phenomenon {is} akin to the pricing of  credit correlation  instruments  where it has been observed that instantaneous correlation between hazard rates {is} unable to generate the sufficient  level of dependence  to hit the market spreads of index tranches (see{, for example,} \cite{BMP, BrMerc, cherubini}).
}

{
Using the  latter modelling approach we will  show how the introduction of  jump--to--default effects achieves   a  much stronger FX/Credit dependence  than  correlated Brownian motions. 
In particular{,}  the addition of FX jump{s} allows to recover both the EUR and the USD spreads {(see} the results presented in Section \ref{sec:backtest}{).}
Furthermore, we  show a powerful{, yet simple,} way of  extracting the magnitude of currency devaluation upon default  from the CDS market data ({see }Section {\ref{sec:basktextResults}}).
}

{In addition to multi--currency CDSs{,} the  quanto effect in credit modelling finds a natural application in the \emph{credit valuation adjustment} (CVA) space. 
CVA is a{n} adjustment to the fair value of a derivative contract that {accounts for the expected loss due to the counterparty's default}. 
{We refer the interested reader to \cite{BMP} for a comprehensive overview of CVA modelling and to \cite{cherubiniCollateral} for specific discussions about collateral modelling.}
Modelling the dependence between credit and market risk factor{s} is  crucial  to accurately calculate the  CVA charge. 
One of  the main challenges in calculating CVA is the lack of liquid  CDS market data {to calibrate model parameters}.
The calibration and approximation techniques {showed in this paper} to  connect  currency devaluation with multi--currency CDS par--spreads {can as well be applied to CVA modelling --- for example, to better reflect right--way or wrong--way risk.}
The resulting FX/Credit cross modelling improvement is crucial{,} {especially in those cases} where the interaction between the counterparty credit  and the FX {component}  is strong,  i.e. with emerging market credits and systemically relevant counterparties. 
}

{In Section \ref{sec:genFrame}, we show how the introduction of default--driven FX jumps changes the dynamics of the stochastic hazard rate after a measure change.
This happens because, f}rom a mathematical perspective{,}  the FX rate  is a component of the  Radon--Nikodym derivative that links the risk neutral probability measures associated to two {different }currencies.
{As stated by Girsanov Theorem (see 
\iftoggle{submission}{, for example, \cite{JeanYor}}{Appendix \ref{app:girsanov}}), the dynamics of the {compensated} default process under different risk--neutral measures differ in their drift component.
Such drift depends on the quadratic covariation between the FX rate and the default process (and it is zero when such covariation is null) and can be interpreted as the stochastic hazard rate of the reference entity.
}

{The above result is strongly linked to another aspect of FX rate modelling, which we will refer to as FX symmetry throughout  this document (see the discussion in section \ref{sec:2measures}). 
Consistency between an FX rate process and its reciprocal is not guaranteed under every possible distributional assumptions made on its dynamics. 
For example, in case {of} stochastic volatility FX modelling, {the} reciprocal FX rate  would  not necessarily have the same dynamics that one would expect given that  the reciprocal FX rate is also a Radon--Nikodym derivative. 
For geometric Brownian motions{,} however, this consistency is guaranteed. 
Due to the change in the hazard rate  in the second pricing measure  induced by the   jump--to--default feature of  the FX rate/Radon--Nikodym derivative process, we prove in section \ref{sec:fxSymmetry} that the symmetry is preserved also for {our} specific FX model.}

\subsection{{Previous }{L}iterature }

{We refer to \cite{rutPDE} for an overview of the general problem of deducing {a} PDE to price defaultable claims and to \cite{bieleckiCDShedge} for the specific problem of CDS hedging in a reduced--form framework.}

For an introduction to the joint modelling of credit and FX in a reduced--form framework with application to Quanto--CDS pricing, we refer to \cite{Schon, FxBofa}. 
\cite{Schon} propose the idea to link FX and hazard rate by considering a jump--diffusion model for the FX--rate process where the jump happens at the default time. 
{Differently from the present work,} no explicit derivation of the PDE is presented, as the focus is on affine processes modelling.

 The same idea is presented and developed in \cite{FxBofa}. In that work it is shown how to calculate quanto--corrected survival probabilities using a PDE--based approach. In order to do that, the author deduces a Fokker--Planck equation for the joint distribution of FX and hazard rate. 

The approach we present in Section \ref{sec:indepModel} below is based on the same {Jump--to--Default} framework as the one used in the references above. In our case, however, the calculation of the quanto--corrected survival probabilities depends on solving a Feynman--Kac equation, the solution of which is a price, while in \cite{FxBofa} a probability density distribution was calculated. 
{At implementation level, the difference between the two approaches lies in the fact that in the latter case an additional integration step would be required to calculate a price.}
Additionaly, the way we work out our main pricing equation makes clear what instruments and in what amounts one would need to effectively implement a delta--hedging strategy.


An algorithm using a fixed--point approach has been recently proposed to calculate CVA in \cite{kimleung}.

The techniques showed in this paper seem particularly relevant for long--maturity trades, where the effects of idiosyncratic jump--to-default components on counterparty risk can be more pronounced and where, therefore, they can have a big impact  on wrong ray risk estimation. 
For a relevant example of CVA calculations related to long--maturity trades, we refer to \cite{biffispitotti}, where the cost of CVA and collateralization are calculated for longetivity swaps.

{The use of L\'evy processes with local volatility to price options on defaultable assets has been recently explored in \cite{pascucciDefault}, where a family of asymptotic expansions for the transition density of the underlying is derived.}
{Differently from the approach presented in this paper, in that case a single stochastic process drives both the default intensity and the option's underlying.
On the other hand, being able to account for the implied volatility skew is feature currently missing from the framework presented in Section \ref{sec:indepModel} and that will be explored in future works.}

{
With respect to the Republic of Italy's test case that is presented in the results' section \ref{sec:results}, we note that the Euro--area situation presents interesting problems that go beyond the mere credit--FX interaction  which is the focus of the present work. An additional layer of complexity is provided in this case by the interconnectedness between the credit {risk} of the different currencies. 

Empirically, Germany  quanto CDS basis tends to be more pronounced than the Greece one (see \cite{pykhtinRisk}), reflecting  higher  correlation between EUR/USD and Germany hazard rate of default and higher EUR/USD devaluation upon   Germany default.  

}

%

\subsection{Quanto CDS}\label{sec:mechanics}

Quanto CDS are designed to provide protection upon default of a certain entity in a given currency. There are cases, like for sovereign entities or for systemically important companies, when an investor might prefer to buy protection on a currency other than the one in which the assets of the reference entity are denominated. A typical reason for entering this type of  trades would be to avoid the FX risk linked to the devaluation effect associated to the reference entity's default. 

Alternatively, protection might be needed in a different currency from the one in which the assets of the reference entity are denominated  because it serves as a hedge on a security denominated in that specific currency. 

{The {discounted} cashflows of the premium leg, $\Pi^{\textrm{Premium}}${,} are given {(as seen from the protection seller's perspective)} by}
\begin{equation}\label{eqn:prem}
\Pi^{\textrm{Premium}} = S^c \sum_{i=0}^{N}\1{\tau>T_i}D_0^{ccy}(T_i)
\end{equation}
{where}
\begin{itemize}
\item
{$(T_0,\dots, T_N)$ is the set of quarterly spaced payment times;}
\item
{$D_t^{ccy}(T)$ is {the stochastic discount factor for} currency $ccy$ at time $t$ {for maturity $T$};}
\item
{$S^c$ is the contractual spread;}
\item
{$\tau$ is the default time of the reference entity}.
\end{itemize}
{The protection leg is made of a single cash flow, $\Pi^{\textrm{Protection}}$, paid upon default of the reference entity on a reference obligation}{:}
\begin{equation}\label{eqn:prot}
\Pi^{\textrm{Protection}} = \lgd\1{\tau\leq T_N}  D_0^{ccy}(\tau),
\end{equation}
{where}
\begin{itemize}
\item
{$\lgd$ is the \emph{loss given default} related to the contract.}
\end{itemize} 
{The spread {$S^c$} that makes the expected value of the cash--flows in Eq \eqref{eqn:prem} equal to the expected value of the cash--flow in Eq \eqref{eqn:prot} is referred to as \emph{par--spread} and we will usually use $S$ to denote it.}
{The existence of CDSs on the same reference entity whose premium and protection cashflows are paid in different currencies creates a basis spread between the par--spreads of these contracts.}
{Figure \ref{fig:cdsScheme} provides a schematic representation of two possible contracts settled in two different currencies.}

\begin{figure}
\begin{tikzpicture}[>=stealth,->,shorten >=2pt,looseness=.5,auto]
\matrix [matrix of math nodes,
column sep={3cm,between origins},
row sep={1.6cm,between origins},
nodes={circle, draw, minimum size=7.5mm},
,ampersand replacement=\&]
{
 |(AEUR)| A \& |(BEUR)| B\& |(AUSD)| A \& |(BUSD)| B\\
};
\begin{scope}[every node/.style={font=\small\itshape}]
\draw[blue] (AEUR) to [bend left] node [midway] {$S^{ccy1}(ccy1)$} (BEUR);
\draw[blue,dashed] (BEUR) to [bend left] node [midway] {LGD(ccy1)} (AEUR);
\draw[red] (AUSD) to [bend left] node [midway] {$S^{ccy2}(ccy2)$} (BUSD);
\draw[red,dashed] (BUSD) to [bend left] node [midway] {LGD(ccy2)} (AUSD);
\end{scope}
\end{tikzpicture}
\caption{Protection on a given reference entity can be bought by A from B in different currencies. 
The stream of payments in Eq \eqref{eqn:prem} is indicated by the solid arrow, while the dashed arrow is used for the contingent payment in Eq \eqref{eqn:prot}.
{The $LGD$ payment, albeit settled in different currencies, is the same percentage of the notional in the two contracts.}}
\label{fig:cdsScheme}
\end{figure}
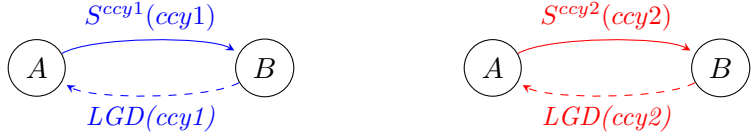

We refer to \cite{jpmqcds} {and references therein} for an overview on quanto CDS market{s} and for a thorough exposition of the rules governing these contracts. We note here that
\begin{itemize}
\item
the standard contracts for sovereign CDS are denominated in USD. This means in particular that for countries of the EUR zone, like Italy, Greece or  Germany, the modeling set up to use when including a devaluation approach is the one detailed in Section \ref{sec:workingSection};
\item
upon default of the reference entity, a common auction sets the loss given default (LGD). The LGD so defined is valid for all the CDSs, irrespectively of the currency they are denominated in. 
\end{itemize}

\subsection{Main {{C}ontribution}}

In this paper, we derive the pricing equations for quanto CDS in different models within the reduced--form framework. In doing so, we show two of the main mechanisms to model dependence between the credit and the FX rate component. 
We will refer to the currency in which the CDSs written on the reference entity are more liquid as to the ``liquid currency'', {that will also define the risk neutral measure used for pricing. 
We will assume that CDSs in a different currency from the liquid one exist and we will refer to this second currency as the ``contractual currency''}. 
In particular, we discuss the mathematical implications of the introduction of a devaluation jump on the spot FX rate between the {contractual} currency and the liquid corrency, both on the pricing equations and on the main risk factors. More in detail:
\begin{enumerate}
\item
in Proposition \ref{prop:one} we show that{,} if we assume for the  FX rate defining the value of one unit of {contractual} currency in the {liquid} currency a dynamics 
\begin{equation}\label{eqn:fxX}
\dif Z_t = \mu^Z Z_t \dif t + \sigma Z_t \dif W_t + \gamma^ZZ_{t-}\dif D_t,\quad Z_0 = z,
\end{equation}
where  $D_t = \1{\tau<t}$ is the default process, then
the hazard rates in the two currencies are linked by
\begin{equation}\label{eqn:lambdaGamma}
 \hat\lambda_t = (1+\gamma^Z)\lambda_t;
\end{equation}
where $\hat \lambda$ is the hazard rate in the measure linked to the {contractual} currency and $\lambda$ is the hazard rate in the currency linked to the {liquid} currency.

{An important corollary of this result is that, in cases where CDS par--spreads can be approximated through the relation $S = (1-R) \lambda$, a similar result holds for par--spreads, too}{:}
\begin{equation}\label{eqn:lambdaGamma}
\hat S = (1 + \gamma^Z) S.
\end{equation}
{We show in Section \ref{sec:results} how such an approximation is applicable to Republic of Italy's par--spreads in the time period ranging from 2011 to 2013;}
\item
in Section \ref{sec:fxSymmetry}	 we show that if we assume for the  FX rate the dynamics given in Eq \eqref{eqn:fxX}, then
\begin{enumerate}[i)]
\item
by no--arbitrage considerations, the drift of $\process Z$ is given by 
\begin{equation*}
\mu^Z= r - \hat r - \gamma^Z \lambda_t(1-D_t);
\end{equation*}
where $r$ is the risk--free rate in the {pricing} measure {linked to the liquid currency} and $\hat r $ is the risk--free rate in the {contractual} measure.
Alternatively, by symmetry considerations,  we could model the {reciprocal} FX rate $X = \sfrac 1 Z$ using the same type of jump--diffusion process 
\begin{equation*}
\dif X_t = \mu^X X_t\dif t - \sigma X_t \dif W_t +  \gamma^XX_{t-}\dif D_t,\quad X_0 = \frac 1 z,
\end{equation*}	
and in this second case we would obtain a drift given by 
\begin{equation*}
\mu^X = \hat r - r- \gamma^X \hat \lambda_t(1-D_t),
\end{equation*}
where
\begin{equation*}
\gamma^X = -\frac{\gamma^Z}{1+\gamma^Z};
\end{equation*}
\item
in Proposition \ref{prop:FX}  we show that the no--arbitrage dynamics implied for $\process X$ is {of the same type as} the no--arbitrage dynamics of $\process Z$. 
This is a result that might not hold in general, for example when stochastic volatility is also included{, or with a price--inhomogeneous local volatility model like CEV}; 
\end{enumerate}
\item
in Proposition \ref{prop:gammaK}  we show an approximated formula, valid for short maturity CDSs,  to estimate the devaluation rate paramenter $\gamma$ and we present numerical results corroborating it in Section \ref{sec:results}.
\end{enumerate}
We study in detail the case of the currency basis spread for CDSs written on Italy in the period 2011--2013 providing, for each day in that time range, the results of the calibration of a model that includes a jump--to--default effect on the FX rate. We show the calibrated parameters and how the calibrated model parameters produce estimates which are consistent with the approximated formula in Eq \eqref{eqn:lambdaGamma}.

\section{Model Description}
 \label{sec:indepModel}
Our modelling framework for credit {risk} falls into the reduced--form approach and, as such, describes not only the evolution of survival probabilities, but also the default event. 

{In Section \ref{sec:ccys} we introduce some definitions concerning the role of different currencies involved in pricing a quanto CDS.}

In Section \ref{sec:2measures} we introduce the general framework that we will refer to to work with two financial markets. In Section \ref{sec:qCDS} we introduce some useful formulae and definitions to price multi currency credit default swaps.

In Section \ref{sec:BK}	we will model a stochastic hazard rate as a {exponential Ornstein--Uhlenbeck} process and the FX rate as a Geometric Brownian Motion (GBM) and we will consider the two driving diffusions to be correlated. 




In Section \ref{sec:genFrame} we present our proposal to embed a factor devaluation approach onto the FX rate dynamics. This provides a way to extend the model shown in Section \ref{sec:BK}.

We {begin by considering} a probability space $(\Omega, \mathcal F, \mathbb Q, \process{\mathcal F})$ satisfying the usual hypotheses. 
In particular $\process{\mathcal F}$ is a filtration under which the dynamics of the risk factors are adapted and under which the default time of the reference entity is a stopping--time. 
Depending on the specific examples, we will also consider spaces  with a different equivalent measure, for example the risk neutral measure associated to the {liquid} {money} market or the risk neutral measure associated to the {contractual currency money} market.

{
We will be {using} {a} Cox process model for the default component and we will refer to the stochastic intensity of the default event simply as \emph{hazard rate} or \emph{intensity}, using the two terms interchangeably.
}

{Unlike} the usual approach followed in the so called  ``reduced--form'' framework for credit risk modelling (see \cite{lando}, \cite{BrMerc}), we do not introduce a second filtration with respect to which only the stochastic processes driving the market risk--factors are measurable\footnote{The total filtration $\process{\mathcal F}$, inclusive of market and default risk, is the only filtration we will consider (that is called $(\mathcal G_t)$ in \cite{BrMerc}). We note that the practical reason for considering this second filtration is because that allows to apply theoretical results developed to price interest rates derivative to credit risk derivatives pricing. Due to the specific model choices we make in the following, however, this would not present any real advantage, while, as shown in sections \ref{sec:bk_pde} and \ref{sec:pEqJumps}, working with a single filtration gives us the  possibility to calculate the quanto adjustment using a PDE approach.}.
\subsection{{The {R}oles of the {C}urrencies}}\label{sec:ccys}
{In this section we set up some definitions concerning the role of the currencies that will be used in our modelling approach.}

{For  of any quanto CDS pricing, we will be considering the following two relevant currencies:}
\begin{itemize}
\item
{\emph{contractual currency} --- This currency is a contract's attribute: it is the currency in which both premium leg and protection leg payments are settled. When considering applications to quanto CDS, for a given reference entity, CDSs are available in at least two different contractual currencies;}
\item
{\emph{liquid currency} --- This is the contractual currency of the most liquidly traded CDS on a given entity. It is used to define a  risk--neutral measure used to price and calibrate the model.}
\end{itemize}


{We list here two examples to illustrate the use of the contractual and liquid currencies. 
}
\begin{enumerate}
\item
{the pricing in USD--measure of a CDS on Republic of Italy settled in EUR;}
\item
{the pricing in USD--measure of a CDS on Republic of Italy settled in USD.}
\end{enumerate}

\begin{table}
\begin{tabular}{  l l l}
\toprule
&Test case 1& Test case 2 \\
\midrule
Contractual currency & EUR& USD\\
Liquid currency & USD & USD\\
\bottomrule
\end{tabular}
\caption{Currencies involved in the priicing of the test cases detailed in Section \ref{sec:ccys}.}\label{tbl:ccys}
\end{table}
{We specified the values of the two currencies for each of these test cases in Table \ref{tbl:ccys}. 
We chose the  test cases so that for all of them USD is the the liquid currency, but this is not necessarily true for {all} CDS available in multiple currencies.  
It is worth noting that the test case 2 can be priced using a usual single currency approach. 
Test cases 1 and 2 will be used in Section \ref{sec:backtest} to illustrate the capability of the model specified in Section  \ref{sec:workingSection} to explain the currency basis observed in the market.}

\subsection{Two {M}arkets {M}easures}\label{sec:2measures}

In this section we summarize known results {about} change of measure in presence of FX effects. This is mostly done to establish notation and set the scene for {the following} original developments. 

Let us consider {the two economies linked to the liquid currency and to the contractual currency, respecively. Let us also} consider the {corresponding} money market account{s} as the numeraire{s} for both the economies. We will use a hat $\hat{\ }$, to denote variables in the {contractual--currency} economy, so that, for example, the two numeraires are $(B_t, t\geq 0)$ for the {liquid--currency} economy and $(\hat B_t, t\geq 0)$ for the {contractual--currency} economy.
The money market account{s'} dynamics are given by
\begin{align}
\dif B_t &= r_tB_t\dif t, \quad  B_0 = 1,\label{eqn:B}\\ 
\dif \hat B_t &= \hat r_t\hat B_t\dif t, \quad  \hat B_0  = 1 \label{eqn:Bhat},
\end{align}
where $\process{r}$ and $\process{\hat r}$ are the stochastic processes describing the short rates in the two economies.

Let us also consider an exchange rate $(X_t, t\geq 0)$ between the currencies of the two economies. 
{$X_t$} is defined {as the price of one unit of the liquid currency expressed as} units of the foreign currency in a spot exchange at time $t$.

We are interested in finding an expression for the Radon--Nikodym derivative that changes the probability measure from $\hat{\mathbb Q}$ to $\mathbb Q$. 
{This} can be worked out by using the Change of Numeraire {technique} and a generic  payoff denominated in the {contractual} currency{, represented by the function {$\hat\phi_T$}}. {To do so, w}e consider{, as said above,} the {contractual--currency} money market account, $\process{\hat B}$,  {as a numeraire for the measure $\hat{\mathbb Q}$, }while for the measure ${\mathbb Q}$ {we still use} the {liquid--currency} money market {account, but with} value  denominated in  the {contractual} currency, $\process{(XB)}$. 
The price of the {contractual} currency payoff $\hat\phi$ {can be expressed in the two measures as:} 
\begin{equation}\label{eqn:chngNum}
 \Ef{t}{\frac{\hat B_t }{\hat B_T } \hat\phi_T} = \E{t}{\frac{B_t X_t}{B_TX_T} \hat\phi_T}.
\end{equation}
{The $\Ef{t}{\cdot}$ expectation on the left--hand side}, on the other hand, can be written as
\begin{equation}\label{eqn:changeNum}
\Ef{t}{\frac{\hat B_t }{\hat B_T }\hat \phi_T} = \Ef{t}{\frac{\hat B_t B_T X_T}{\hat B_T B_t X_t} \frac{B_tX_t}{B_TX_T} \hat\phi_T}
\end{equation}
{and the two expressions above can be used to obtain the Radon--Nikodym derivative that defines the change of measure from $\hat{\mathbb Q}$ to $\mathbb Q$:} 
\begin{equation}\label{eqn:rnDer2}
L_T:=\frac{\dif {\mathbb Q}}{\dif \hat{\mathbb Q}}|_{\mathcal F_T} = \frac{B_TX_T}{B_tX_t}\frac{\hat B_t }{\hat B_T}
\end{equation}

%
In deducing the form of  $\process L$ we started from expected values conditioned on $\mathcal F_t$. 
{Throughout this work, h}owever,  we will mainly be interested {in} expected values conditioned at $\mathcal F_0$ so that for all the applications in the following sections we will be using the formula above with $t = 0$ and $T = t$, namely
\begin{equation}\label{eqn:L}
L_t = \frac{B_t}{\hat B_t}\frac{X_t}{X_0}, \quad L_0 = 1.
\end{equation}
\begin{ass}\label{ass:detRates}
In the following we will be considering deterministic interest rates both for the {liquid--currency} and for the {contractual--currency} economy. 
This means that the money market accounts will be described by
\begin{align}
\dif B_t &= r(t)B_t\dif t, \quad B_0 = 1,\label{eqn:Bdeterm}\\ 
\dif \hat B_t &= \hat r(t)\hat B_t\dif t, \quad \hat B_0 = 1 \label{eqn:BhatDeterm},
\end{align}
in place of \eqref{eqn:B} and \eqref{eqn:Bhat}. To lighten the notation, in most cases we will drop the $t$--dependency for $r(t)$ and $\hat r(t)$ in the following equations.
\end{ass}

The process defined in Eq \eqref{eqn:L} has to be a martingale in the foreign measure. This condition can be used to  determine, together with Assumption \ref{ass:detRates}, the drift of $\process X$. By Ito's {formula}, the dynamics of $\process L$ can be written as
\begin{equation}\label{eqn:rnMartiingale}
\dif L_t = \dif\left( \frac{ B_t}{\hat B_t} \frac{X_t}{X_0}\right) = \frac{B_t}{\hat B_tX_0}(\dif X_t  + r X_t \dif t - \hat r X_t \dif t), \quad L_0 = 1.
\end{equation}
If for example we assume a lognormal dynamics for the FX rate
\begin{equation}\label{eqn:x}
\dif X_t = \mu^X X_t\dif t + \sigma X_t \dif  \hat W_t, \quad X_0 = x_0,
\end{equation}
then asking that $\process L$  in Eq \ref{eqn:rnMartiingale}	is a martingale brings to the {familiar} condition 
\begin{equation}\label{eqn:mu}
\mu^X = \hat r - r.
\end{equation}
\begin{remark}
More generally, the same result holds true in case of a $\process X$ of the type
\begin{equation}
\dif X_t = \mu^X X_t \dif t+ \nu \dif \hat I_t,\quad X_0 = x,
\end{equation}
where $\process{\hat I}$ is a generic $\hat{\mathbb Q}$--martingale.
\end{remark}
An equivalent argument would lead, starting from the {contractual--currency} measure and going to the {liquid--currency} one, to set a drift condition for the process $\process Z$ defined as $Z_t = \frac 1{X_t}$. 
We can define it along the same lines of what was done with $\process X$, as a geometric Brownian motion with a drift to be determined through arbitrage considerations
\begin{equation*}
\dif Z_t = \mu^Z Z_t \dif t + \sigma^Z Z_t \dif W_t, \quad Z_0 = z.
\end{equation*}
The Radon--Nikodym measure in this case would be given by
\begin{equation}\label{eqn:rnDer4}
\hatL_t= \frac{Z_t \hat B_t}{Z_0 B_t},\quad \hatL_0 = 1.
\end{equation}
Requiring that $\process{L^{l\rightarrow c}}$ {has to} be a martingale under the {liquid--currency} measure, would set the drift term as 
\begin{equation}\label{eqn:muZ}
\mu^Z= r -\hat r.
\end{equation}

\begin{remark}[Symmetry]\label{rem:symmetry}

Alternatively, one could deduce the dynamics for $\process Z$ in $\mathbb Q$ starting from $\process X$, whose dynamics is known in $\hat{\mathbb Q}$. 
By applying Ito{'s formula} {to the process given by $Z_t = f(X_t)$ where $f(x) = \sfrac 1 x$},  it would be possible to deduce the dynamics of $\process Z$ in $\hat{\mathbb Q}$. Once its dynamics is known, the form of the driving martingales under $\mathbb Q$ can be worked out using   Girsanov Theorem.  Under the log--normal dynamics chosen for the FX rates, this latter approach and the one starting from the Radon--Nikodym derivative in Eq \eqref{eqn:rnDer4}  lead to the same result. A detailed calculation in case the dynamics of the FX rate is subject also to jump--to--default effect, is presented in Section \ref{sec:fxSymmetry} {below}.
\end{remark}

%

There are cases, for example stochastic volatility FX rate models, where starting from {a different} specification of the FX rate  can make a difference, because the consistency between the arbitrage--free dynamics obtained under the  two different specifications is not guaranteed. In these models, if one starts from $X$ as a primitive modelling quantity, and then implies the distribution of $Z$ at some time $t$  from the law of $X_t$, what will be obtained can be a different distribution from the one that one would have had by starting from $Z$ as a primitive modelling quantity based on the same {dynamical} properties as $X$.

{In applications to quanto CDS pricing, where the FX rate is used in Eq \eqref{eqn:chngNum}},{and where, depending on the circumstances, we might be interested in pricing or calibrating either under the liquid--currency measure or under the contractual--currency measure,} there is a degree of arbitrariness in using one specification 
or the other. Having consistency between the two specifications is a desirable property to avoid  results that depend on the aforementioned choice.

\subsection{Modeling {F}ramework for the Quanto CDS {C}orrection}\label{sec:qCDS}


{In this section we derive model--independent formulas to price contingent claims where contractual currency is different from the {liquid} currency{ used to define the pricing measure}. 
In the next sections we will show the application of these formulas under different dynamics assumptions for the main risk factors.}  

{Let us}  start by calculating the value of a defaultable zero--coupon bond{; it} will be then used as a building block to calculate CDS values.  To do so, we choose a payoff function $\hat\phi_T = \1{\tau>T}${ in} Eq~\eqref{eqn:changeNum} { and write }
\begin{equation}\label{eqn:changeMeasure}
\hat V_t(T) = \Ef{t}{\frac{\hat B_t}{\hat B_T}\ \1{\tau>T} } = \E{t}{\frac{\hat B_t}{\hat B_T}\ \1{\tau>T}\  \frac{\dif \hat{\mathbb Q}}{\dif \mathbb Q} }.
\end{equation}


Using the Radon--Nikodym derivative in \eqref{eqn:rnDer2}, the price of the contingent claim in the {contractual currency} economy {can be calculated by taking the expectation in the liquid currency economy:}
\begin{equation*}
\hat V_t(T) = \frac{B_t}{Z_t} \E{t}{\frac{Z_T}{B_T}\1{\tau>T}}.
\end{equation*}
Under Assumption \ref{ass:detRates} the above can be rewritten as
\begin{equation}\label{eqn:Vhat}
\hat V_t(T) = \frac{B(t,T)}{Z_t}\E{t}{Z_T \1{\tau>T}},
\end{equation}
where $B(t,T) = B_t/B_T$ is the discount factor from time $T$ to time $t\leq T$.

It might be useful\footnote{{Mostly for computational reasons because such definition would easily allow CDS pricers defined for single currency calculations to be re--used for quanto CDS pricing.}} to  define the foreign currency survival probabilities as
\begin{equation}
\hat p_t(T) := \frac{\hat V_t(T)}{\hat B(t,T)}.
\end{equation}

Let us now consider {the price, expressed in liquid currency, of the defaultable zero--coupon bond settled in the contractual currency, $U$.
This is given by:}
\begin{equation}
U_t(T) = \hat V_t(T) {Z_t} = B(t,T)\E{t}{Z_T\1{\tau>T}}. 
\end{equation}
Being the {$\mathbb Q$--}price of a tradable asset, {the drift of the process} $\process U$ has to be {given by $r(t)U_t \dif t$}. Therefore, we can  write a Feynman--Kac equation to calculate $U_t(T)$. Once $U_t(T)$ is known, $\hat p_t(T)$ can be calculated as 
\begin{equation}
\hat p_t(T) = \frac{U_t(T)}{Z_t \hat B(t,T)}.
\end{equation}

\subsection{A diffusive correlation model: {{e}xponential OU / GBM}} \label{sec:BK}
In this section we present a specific model to calculate $U$. 
{W}e will be working with a hazard rate process {and a FX rate process} which {are} defined and calibrated in the {liquid} measure. 

Let us denote by {$\process{\lambda}$}  a stochastic process given by {$ \lambda_t = e^{Y_t}$} where $(Y_t, t \geq 0)$  is an Ornstein--Uhlenbeck process defined as the solution of
\begin{equation}\label{eqn:Y}
\dif Y_t = a(b - Y_t)\dif t + \sigma^Y \dif  W^{(1)}_t,\quad Y_0 = y,
\end{equation}
{where the parameters $(a, b, \sigma^Y, y)\in \mathbb R^+\times\mathbb R^+\times\mathbb R^+\times\mathbb R^+$.}
Let us also consider a GBM process  for the FX rate 
\begin{equation}\label{eqn:lognFX}
\dif Z_t = \mu^Z Z_t \dif t + \sigma^{Z} Z_t \dif W^{(2)}_t\quad Z_0 = z,
\end{equation}
where $\mu^Z$ is set by no arbitrage considerations and it is given in this case by Eq \eqref{eqn:muZ}{, and where $(\sigma^Z, z)\in\mathbb R\times\mathbb R^+$}.

The dependence between FX and credit can be specified in this model th rough the instantaneous correlation {(quadratic covariation)} between the two driving Brownian motions, $\rho\in[-1,1]$,
\begin{equation*}
\mathbb \dif\ \langle  W^{(1)} ,W^{(2)} \rangle_t = \rho \dif t.
\end{equation*}
From the results in Section \ref{sec:2measures}, the FX rate in the opposite direction to $Z$, that is $X = \sfrac 1 Z$ follows a dynamics given by
\begin{equation}
\dif X_t = \mu^X X_t \dif t + \sigma^{X} X_t \dif \hat W^{(2)}_t,\quad X_0 = x,
\end{equation}
with $\mu^X$ given by Eq \eqref{eqn:mu} and $\sigma^X = - \sigma^Z$. 

Let finally $(D_t, t \geq 0 )$ be the default process $D_t = \1{\tau < t}$.


\begin{remark}
Due to the symmetry relation holding for FX rates that are modeled as geometric Brownian motions that was stated in Remark \ref{rem:symmetry}, it does not matter if we choose to model $\process Z$,  or $\process X$, as the two dynamics are consistent.
\end{remark}

\begin{remark}
The choice of the ({exponential OU} and GBM) dynamics  has been mainly driven by the need  for the hazard rate process to stay non negative. However, different hazard rates dynamics, possibly with local volatilities, can easily be accounted for using the same framework presented below as far as they only {driven by} Wiener processes and no jump processes are involved. 
Extensive literature has been produced on the use of square root processes for default intensity, mostly due to their tractability in obtaining closed form solutions for Bonds, CDS and CDS options, see for example \cite{brigo1} and \cite{brigo2}, where exact and closed form calibration to CDS curves is also discussed. 
For the FX rate dynamics, instead, there is no such freedom of choice as the  drift is given by no--arbitrage conditions, and introducing local{ or stochastic} volatilities might break the symmetry relation between the FX rate and its reciprocal. 
\end{remark}

\subsubsection{{Hazard {R}ate's {D}ynamics in the $\hat{\mathbb Q}$ {M}easure}}
We are assuming that the hazard rate process dynamics is known in {${\mathbb Q}$}. {Knowing the Radon--Nikodym derivative between measure {${\mathbb Q}$} and measure {$\hat{\mathbb Q}$} would allow us to write the dynamics of the hazard rate in {$\hat{\mathbb Q}$}. 
That can be obtained by using } Girsanov's Theorem\iftoggle{submission}{}{ (see \ref{theo:girsanov})}{, from which}
\begin{equation}
\dif \hat W^{(1)}_t = \dif W^{(1)}_t - \frac{\dif\ \langle W^{(1)}, Z \rangle_t}{Z_t} = \dif W^{(1)}_t - \rho \sigma^{Z} \dif t
\end{equation}
{so that}
\begin{equation}\label{eqn:Y}
\dif Y_t = a(b - Y_t) \dif t- \sigma^Y \rho \sigma^{Z} \dif t + \sigma^Y \dif  \hat W^{(1)}_t .
\end{equation}

\subsubsection{Pricing Equation}\label{sec:bk_pde}
{In this section we deduce a pricing equation to calculate the value of $U$. We follow the approach used in \cite{rutPDE}.}
Given the strong Markov property of all the processes defined so far, $U_t(T)$ can be expressed as a function of $t$,  $Z_t$, $Y_t$ and $D_t${. L}et{ us} denote its value at $t$ for $Z_t = z$, $Y_t = y$ and $D_t = d$ by  $f(t, z, y, d)$. $f$ is a function depending on both continuous and jump processes, and its Ito differential can be written as (see 
\iftoggle{submission}{, for example, \cite{JeanYor}}{Appendix \ref{app:ito}})
\begin{multline}\label{eqn:itoBK}
\dif f_t = r f\dif t + \partial_tf \dif t+ \partial_z f \left(  \mu^Z z \dif t + \sigma^{Z} z \dif W^{(2)}_t \right) + \partial_y f \left(  a(b - Y_t) \dif t
+ \sigma^Y \dif  W^{(1)}_t  \right) \\
+ \frac 1 2 \left( \sigma^{Z} z\right)^2 \partial_{zz}f \dif t + \frac 1 2 \left( \sigma^{Y} \right)^2 \partial_{yy}f \dif t + \rho \sigma^Z \sigma^Y z \partial_{zy} f \dif t +
\Delta f \dif D_t
,
\end{multline}
where, with some abuse of notation, we have defined the jump--to--default term as
\begin{equation}\label{eqn:Deltaf}
\Delta f := f(t, Z_{t-}+\Delta Z_t, Y_t, D_{t-} + \Delta D_t) - f(t, Z_{t-}, Y_t, D_{t-}).
\end{equation}
A compensator for $\process D$ in the measure $\mathbb Q$ is defined as the  process $\process A$ such that $D_t - A_t$ is a $\mathbb Q$--martingale {with respect to {$\process{\mathcal F}$}. The compensator for $\process D$ is given by (see Lemma 7.4.1.3 in  \cite{JeanYor})
\begin{equation}\label{eqn:compensator}
\dif A_t = \1{\tau>t}\lambda_t\dif t .
\end{equation}
{We define the resulting martingale as $\process M$. It is given by}
\begin{equation}
M_t = D_t - A_t.
\end{equation}
Consequently, the compensator of the last term in Eq \eqref{eqn:itoBK} can be written as 
\begin{equation}
\1{\tau>t} e^{Y_t}\Delta f,
\end{equation}
which, {conditional on $\mathcal F_t$}, $D_t = d$, $Z_t =z$, and $Y_t = y$, is equal to
\begin{equation}
(1-d)e^{y}\left(f(t,z(1 + \gamma^Z),y,1)- f(t,z,y,0)\right)\dif t.
\end{equation}
It is possible to write a Feynman--Kac type PDE to compute the value of $U_t(T)$.  Indeed $\process{U}$  is a {$\mathbb Q$--}price and, as such, {it must locally grow at the rate $r$}. Therefore, its drift must satisfy the following equation
\begin{multline*}
 \partial_tf  
+  \mu^Z z  \partial_z f    +  
 a(b - Y_t) 
 \partial_y f    
+ \frac 1 2 \left( \sigma^{{Z}} z\right)^2 \partial_{zz}f  \\
+ \frac 1 2 \left( \sigma^{Y} \right)^2 \partial_{yy}f  + \rho \sigma^Z \sigma^Y z \partial_{zy} f  + e^{y}(1-d) \Delta f  = 0,
\end{multline*}
where the explicit dependence of $f$ on the state variables $(x,y,t,d)$ has been omitted for clarity of reading.
If it wasn't for the last term, this would be the typical PDE for {default}--free payoffs. Incidentally, this jump--to--default term is also the only term of the equation where the values $f(t, z, y, 0)$ and $f(t, z, y, 1)$ appear together. In fact, by conditioning first on $d = 1$ and then on $d = 0$ we can decouple the two functions 
\begin{align}
u(t,z,y) &:= f(t, (1+\gamma^Z) z, y, 1),\\
v(t,z,y) &:= f(t, z, y, 0)
\end{align}
and calculate them by solving iteratively two separate PDE problems. We first solve for $u$, as for $d = 1$ the last term does not appear in the equation, and, once $u$ has been calculated,  we use it to solve for $v$. Final conditions for the two functions  are respectively given by
\begin{align}
v(T,z,y) = f(T,z,y, 0) &= z;\\
u(T, z, y) = f(T,z, y, 1) & = 0.
\end{align}

The PDE problem that must be solved to obtain $u$ is then given by
\begin{align}\label{pde:BKu}
 \partial_tu  & =
-   \mu^Z z  \partial_z u    - 
 a(b - y)
 \partial_y u    
- \frac 1 2 \left( \sigma^Z x\right)^2 \partial_{zz}u  \nonumber \\
&- \frac 1 2 \left( \sigma^{Y} \right)^2 \partial_{yy}u  - \rho \sigma^Z \sigma^Y z \partial_{zy} u   \\
u(T, z, y) & = 0.
\end{align}
The solution to this problem is $u\equiv 0$, therefore in this case one can  solve directly the PDE for $v$, which is then given by
\begin{subequations}\label{pde:BKv}
\begin{align}
 \partial_tv  & =
-   \mu^Z z  \partial_z v    - 
 a(b - y)
 \partial_y v    
- \frac 1 2 \left( \sigma^Z x\right)^2 \partial_{zz}v  \nonumber \\
&- \frac 1 2 \left( \sigma^{Y} \right)^2 \partial_{yy}v  - \rho \sigma^Z \sigma^Y z \partial_{zy} v  + e^{y}\  v \\
v(T, z, y) & = z.
\end{align}
\end{subequations}


\begin{remark}[Interpretation of $u$ and $v$]
The functions $u$ and $v$ account for the pre--default and post--default value of a derivative with payoff $\phi(x, y, d)$. 
The price of this derivative can be written as
\begin{equation}
V_t =\1{\tau>t-}\mathbb E_t\left[ \phi(X_T,Y_T,D_T) | X_t = x, Y_t = y, D_t = d\right],
\end{equation}
where, due to the strong Markov property of the processes $\process X$, $\process Y$, and $\process D$, the expected value on the right--hand side can be written as
\begin{equation}
f(t,x,y,d) =\mathbb E_t\left[ \phi(X_T,Y_T,D_T) | X_t = x, Y_t = y, D_t = d\right]. 
\end{equation}
This can be decomposed as $f(t,x,y,d) = \1{d = 1}u(t, x,y) + \1{d = 0}v(t, x, y)$ where
\begin{align}
v(t,x, y)&:=\mathbb E_t\left[ \phi(X_T,Y_T,D_T)| X_t = x, Y_t = y, D_t = 0\right],\\
u(t,x,y)&:=\mathbb E_t\left[ \phi(X_T,Y_T,D_T)| X_t = x, Y_t = y, D_t = 1\right], 
\end{align} 
in fact
\begin{align}
f(t,x,y,d) &= \mathbb E_t\left[ \phi(X_T,Y_T,D_T) | X_t = x, Y_t = y, D_t = d\right] \nonumber \\
&= \1{\tau>t} \mathbb E_t\left[ \phi(X_T,Y_T,D_T)| X_t = x, Y_t = y, D_t = 0\right] \nonumber \\
&\quad + \1{\tau\leq t} \mathbb E_t\left[ \phi(X_T,Y_T,D_T)| X_t = x, Y_t = y, D_t = 1\right] \nonumber \\
&= \1{\tau>t}  v(t, x, y) + \1{\tau\leq t}  u(t, x, y)
\end{align}
as both $\1{\tau>t} $ and $\1{\tau\leq t} $ are measurable in the {$\mathcal F_t$} filtration.
The derivative price can then be written as
\begin{equation}
V_t = \1{\tau>t}v(t, X_t, Y_t) + \Delta D_t u(t, X_t, Y_t),
\end{equation}
where we defined
\begin{equation}
\Delta D_t := \1{\tau>t} - \1{\tau>t-}.
\end{equation}
\end{remark}
%

\iftoggle{submission}{}{
\begin{remark}[Deterministic hazard rate]
Let us consider a deterministic hazard rate $H(t)$ and deterministic rates. In this case, we can still apply the methodology from this section to write down a pricing equation for a risky zero--coupon bond. In this simplified case, however, the pricing equation is an ODE rather than a PDE. 
\begin{align}
\dif f(t, d) &= \partial_t f(t, d) \dif t + (f(t, 1) - f(t, 0))\dif D_t \nonumber \\
&= \partial_t f(t, d) \dif t + (f(t, 1) - f(t, 0))\dif M_t +  (1-d)(f(t, 1) - f(t, 0))H(t)\dif t
\end{align}
Asking for $f$ to be a martingale leads to
\begin{equation}
\partial_t f(t, d) +  (1-d)(f(t, 1) - f(t, 0))H(t) = 0,
\end{equation}
which, by defining $u(t):= f(t, 1)$ and $v(t):=f(t,0)$ with final conditions $u(T) = 0$ and $v(T) = 1$ gives
\begin{align}
\partial_t u(t) &= 0\\
u(T)&=0
\end{align}
whose solution is $u(t)\equiv 0$, and
\begin{align}
\partial_t v(t) &= H(t) v(t)\\
v(T) &= 1
\end{align}
whose solution at $t=s \in [0,T]$ is $v(s) = e^{- \int_s^TH(u)\dif u}$.
\end{remark}
}

\subsection{A Jump--to--Default {F}ramework} \label{sec:genFrame}

The exponential OU--based model described in Section \ref{sec:BK} can be extended by incorporating a devaluation mechanism in the FX rate dynamics. By linking the devaluation to the default event, it is possible to introduce a further source of dependence between {$\process \lambda$} and $\process X$. 
In Section \ref{sec:results} it will be shown that this will prove to be a {more suitable} mechanism to model the basis spread for quanto--CDS. 

{This section is organised as follows: in the first subsections, from Section \ref{sec:toolkit} to Section \ref{sec:fxSymmetry}, we will discuss in general how the dynamics of the risk factors are affected by the introduction of a jump--to--default effect on the FX component. Given that the Radon--Nikodym derivative depends on the FX rate, this change is expected to have an impact on all the risk factors whose dynamics has to be written in a measure different from the one in which they have been originally calibrated and, potentially, on the FX symmetry discussed is in Remark \ref{rem:symmetry}. 
{This is proven to hold true also in this new, more general, framework (see Proposition \ref{prop:FX}).}
In Section \ref{sec:pEqJumps} we will apply the general results from the first subsections to the pricing of quanto CDS.}

\subsubsection{{Risk {F}actors {D}ynamics}}\label{sec:toolkit}

Let us then consider a jump--diffusion process for the FX rate in place of \eqref{eqn:lognFX}, while we will be keeping the same model choice for the hazard rate $\lambda_t = e^{Y_t}$:
\begin{align}
\dif Y_t &= a(b - Y_t)\dif t + \sigma^Y \dif  W^{(1)}_t, \quad Y_0 = y, \label{eqn:hr_jumpFX}\\
\dif Z_t &= \bar\mu Z_t \dif t + \sigma^{Z} Z_t \dif  W^{(2)}_t + \gamma^Z Z_{t-} \dif D_t,\quad Z_0 = z, \label{eqn:jumpFX}\\
\dif\ \langle W^{(1)},  W^{(2)}\rangle_t & = \rho \dif t \label{eqn:corr_jumpFX}
\end{align}
{where, as before, the parameters $(a, b, \sigma^Y, y)\in \mathbb R^+\times\mathbb R^+\times\mathbb R^+\times\mathbb R^+$, $(\sigma^Z,z)\in\mathbb R\times\mathbb R^+$, $\rho\in[-1,1]$, and} where $\gamma^Z\in[-1,\infty)$ is the devaluation/revaluation rate of  the FX process. 
The typical case in which this devaluation factor is used is for reference entities whose default can negatively impact the value of their local currency. As an example, we expect the value of EUR expressed in  USD to fall  in case of Italy's default.

We leave unspecified  the drift term of $\process Z$  and we simply use $\bar \mu$ for it in order to distinguish it from $\mu^Z$. It will be shown in Section \ref{sec:fxSymmetry} that the introduction of the jump term will lead to a result different from Eq \eqref{eqn:muZ} if we want  the process defined in Eq \eqref{eqn:rnDer4} to still be a martingale.

\begin{remark}[Jumps]
The jump term in SDE for jump--diffusion processes can be described equivalently using $\process D$ or the compensated process $\process M$, the effect of using one term or the other  being just a change in the drift term. We prefer using the non--compensated term when introducing the FX process in order to highlight the jump structure and hence the additional source of dependence between the FX and the credit component. On the other hand, the description in terms of the compensated martingale $\process M$ will arise naturally every time the Fundamental Theorem of {Asset Pricing} will be used to derive no arbitrage drift conditions, e.g. when Eq \eqref{eqn:rnMartiingale} is used to deduce Eq \eqref{eqn:compMu} {below} and, as it will be shown in Section \ref{sec:pEqJumps}, to deduce the main pricing equation.

\end{remark}

\

\subsubsection{{Hazard {R}ate's and FX {R}ate's {D}ynamics in $\hat{\mathbb Q}$}}\label{sec:girsanovJumps}

Given the dependence of $\process \hatL$ on $\process D$ {via $\process Z$}, in this case the change of measure modifies not only the expected value of $\process{ W}$, but also the expected value of $\process{ M}$ which was originally given by $\dif M_t = \dif D_t - (1-D_t)\lambda_t \dif t$ in ${\mathbb Q}$. However, Girsanov's Theorem provides the adjustments for each of these processes needed to obtain a martingale in the new measure. 
\begin{subequations}
\begin{align}
\dif \hat W_t &= \dif W_t - \frac{\dif\ \langle W, Z\rangle_t}{Z_t} = \dif W_t - \sigma^Z\dif t, \label{eqn:wienerTransform}\\
\dif \hat M_t &= \dif M_t -  (1- D_t)\gamma^Z \lambda_t \dif t.
\end{align}
\end{subequations}
The Wiener process decomposition in $\hat{\mathbb Q}$ is given by the same formula used in Section \ref{sec:BK}, while we derive the martingale decomposition for $\process D$ as a result of the following
\begin{prop}\label{prop:one}
Let $\process{ M}$ be the martingale associated to the default process $\process D$ in the domestic currency measure
\begin{equation*}
\dif M_t  = \dif D_t - (1- D_t) \lambda_t \dif t,
\end{equation*}
then an application of the Girsanov Theorem allows to write the correspondent martingale in the {foreign} measure $\process{\hat M}$ as 
\begin{align}
\dif \hat M_t &= \dif  M_t - \frac{\dif\ \langle M, \hatL\rangle_t}{\hatL_t} = \dif M_t - {\dif\ \langle D,  \gamma^Z D\rangle_t} \nonumber\\
&= \dif M_t -  (1- D_t)\gamma^Z \lambda_t \dif t \\
&= \dif D_t -  (1- D_t)(1+\gamma^Z) \lambda_t \dif t  \label{eqn:newLambda}
\end{align}
where the dynamics of {$\process{\hatL}$} is defined by Eq \eqref{eqn:rnDer4} and Eq \eqref{eqn:jumpFX}.
Eq \eqref{eqn:newLambda} states that the intensity of the Poisson process driving the default event in the foreign currency is given by 
\begin{equation}\label{eqn:hrmc}
\hat\lambda_t:= (1+\gamma^Z) \lambda_t
\end{equation}
\end{prop}

\begin{proof}
Integration by parts  gives
\begin{align*}
\dif\ (\hat M_t \hatL_t) & = \hatL_t \dif \hat M_t + \hat M_t \dif \hatL_t + \dif\ [\hat M,\hatL]_t \\
& = \hatL_t \dif \hat M_t + \hat M_t \dif \hatL_t + \gamma ^Z \hatL_t \dif D_t \\
& = \hatL_t (\dif M_t -  (1- D_t)\gamma^Z \lambda_t \dif t )+ \hat M_t \dif\hat  L_t + \gamma^Z \hatL_t \dif D_t \\
& = \hatL_t \dif M_t + \hat M_t \dif \hatL_t + \gamma^Z \hatL_t \dif M_t
\end{align*}
so the process $\process{(\hatL\hat M)}$ is a martingale in the domestic measure as it can be written as a sum of stochastic integrals on local martingales. As  a consequence, the process $\process{\hat M}$ is a local martingale in the foreign measure.
\end{proof}

\begin{remark}[CDS par--spreads approximation]\label{rem:approx}
In all the cases where the well known approximation 
\begin{equation}\label{eqn:trnglAppr}
\lambda \approx \frac{S}{1-R}
\end{equation}
between hazard rates, CDS par--spreads, $S$, and recovery rates, $R$, holds, the relation in Eq \eqref{eqn:hrmc} can be written in terms of CDS par--spreads rather than hazard rates as
\begin{equation}\label{eqn:parSPrApprox}
\hat S= (1+\gamma^Z) S.
\end{equation}
This happens, for example, where the hazard rate is constant in time and when the  premium leg's cash-flows can be approximated by a stream of continuously compounded payments{ (see \cite{BrMerc})}. 
\end{remark}

\subsubsection{Hazard Rates {D}ynamics in the {T}wo {M}easures}

As shown by Proposition \ref{prop:one}, the hazard rate's magnitude changes depending on whether we are pricing a contingent claim in $\hat{\mathbb Q}$ or $\mathbb Q$. 

If we still consider a{n exponential OU} model for the evolution of the hazard rate, the relation obtained in Proposition \ref{prop:one}, {$\hat\lambda_t = (1+\gamma^Z) \lambda_t$} can be translated in terms of the driving processes $\process Y$ and $\process{\hat Y}$ as
\begin{equation*}
Y_t = \log\left( \frac{e^{\hat Y_t}}{1+\gamma^Z}\right)
\end{equation*}
from which 
\begin{equation}\label{eqn:dY}
\dif Y_t = { \dif \hat Y_t}.
\end{equation}
This result {could} be useful when writing the pricing PDE, because the price {could} be calculated as an expectation in the {domestic} measure, while {the} set of stochastic processes {might be} defined in the foreign measure.

\subsubsection{FX {R}ates {D}ynamics in the {T}wo {M}easures and {S}ymmetry}\label{sec:fxSymmetry}

The FX rate in this model is  a jump--diffusion process, whose jumps are given by (see Eq \eqref{eqn:jumpFX})
\begin{equation}
\Delta Z_t = \gamma^Z Z_{t-} \Delta D_t.
\end{equation}
Notice that also this specification of the FX rate is subject to  arbitrage constraints {such that the Radon-Nikodym derivative defined by Eq \eqref{eqn:rnDer4} {be} a martingale}. The condition equivalent to Eq \eqref{eqn:muZ} in the case where the FX dynamics is given by Eq \eqref{eqn:jumpFX} is provided by
\begin{equation}\label{eqn:compMu}
\bar \mu = \mu^Z -\lambda_t \gamma^Z \1{\tau>t} = r - \hat r  - \lambda_t \gamma^Z \1{\tau>t}.
\end{equation}

Despite the introduction of the jump in the FX rate dynamics, the consistency {highlighted in Remark \ref{rem:symmetry}} between $\process X$ and $\process Z$ is maintained. 
From a practical point of view {this} means that we do not need to worry about which FX rate we use, as  one can be  obtained as a transformation of the first one and it is guaranteed to satisfy the no-arbitrage relations for the associated Radon--Nikodym derivative. This is proved in the next 
\begin{prop}[FX rates symmetry under devaluation jump to default]\label{prop:FX}
Let us consider an FX rate process whose dynamics in the domestic measure ${\mathbb Q}$ is specified by Eq \eqref{eqn:jumpFX} and whose  drift is given by Eq \eqref{eqn:compMu}. Then the dynamics of the process $\process X$ where $X_t = \sfrac 1 {Z_t}$ in the foregin measure $\hat{\mathbb Q}$ is given by
\begin{equation}\label{prop:Z}
\dif X_t = (\hat r - r)X_t \dif t - \sigma^ZX_t\dif  \hat W_t^{(2)}  +X_{t-}\gamma^X  \dif \hat M_t, \quad X_0 = \frac 1 z,
\end{equation}
where the devaluation rate for $\process X$ is given by
\begin{equation}\label{eqn:gammaPrimo}
\gamma^X = -\frac{\gamma^Z}{1+\gamma^Z}.
\end{equation}
In particular, \eqref{prop:Z} is such that the Radon--Nikodym derivative defined by Eq \eqref{eqn:L} is a $\hat{\mathbb Q}$-martingale.
\end{prop}

\begin{proof}

See Appendix \ref{app:proofFX}
\end{proof}

Alternatively, a representation where the jumps are highlighted can be used  for the $\hat{\mathbb Q}$-dynamics of $\process X$
\begin{equation}\label{eqn:zJump}
\dif X_t = \left(\hat r - r - (1-D_t)\gamma^X {\lambda_t}\right) X_t \dif t - \sigma^ZX_t\dif  W_t^{(2)}  +X_{t-}\gamma^X  \dif D_t ,  \quad X_0 = \frac 1 z.
\end{equation}

\subsubsection{Pricing Equation}\label{sec:pEqJumps}
\label{sec:workingSection}
In this section, we consider the case where {liquid} currency and pricing currency coincide and are different {from} the {contractual} currency. As discussed in Section \ref{sec:mechanics}, this is the typical setup arising to price {in the USD--market measure} CDSs written on European Monetary Union countries, as the standard currency for them is USD. If one wants to price a EUR denominated contract for such reference entities in the USD measure, one has first to calibrate the hazard rate to USD--denominated contracts and then the pricing can be carried out using the equations derived in this section. This is {also} the procedure followed to produce the results showed in Section \ref{sec:backtest} {below}.

Without loss of generality, we will {study} the case {of} liquid currency {and pricing currency associated} to {the domestic measure} $\mathbb Q$.
\begin{align}
\dif Y_t &= a(b - Y_t)\dif t + \sigma^Y \dif W^{(1)}_t, \label{eqn:yEUR} \\
\dif Z_t &= \bar\mu_Z Z_t \dif t + \sigma^ZZ_t \dif W^{(2)}_t + \gamma^Z Z_t \dif D_t \label{eqn:zEUR}\\
\dif\ \langle W^{(1)}, W^{(2)}\rangle_t & = \rho \dif t \label{eqn:rhoEUR}
\end{align}
{w}ith 
\begin{equation}
\dif M_t = \dif D_t - (1-D_t)\lambda_t \dif t
\end{equation}
so that the no-arbitrage drift is given by (see Eq \eqref{eqn:compMu})
\begin{equation}
\bar\mu_Z = r - \hat r - \gamma^Z (1 - D_t) \lambda_t.
\end{equation}
{An application of the generalized Ito formula (see\iftoggle{submission}{, for example, \cite{JeanYor}}{Appendix \ref{app:ito}}) allows us to write the $\mathbb Q$--dynamics of $\process{U}$. Using $U_t = f(t, Z_t, Y_t, D_t)$:}
\begin{multline*}
\dif f = rf \dif t+ \partial_tf \dif t+ \partial_z f \left( \bar \mu_Z z \dif t + \sigma^Z z \dif W^{(2)}_t + \gamma^Z z \dif D_t \right) + \partial_y f \left(  a(b - Y_t) \dif t+ \sigma^Y \dif  W^{(1)}_t  \right) \\
+ \frac 1 2 \left( \sigma^Z z\right)^2 \partial_{zz}f \dif t +  \frac{1}{2}\left( \sigma^{Y} \right)^2 \partial_{yy}f \dif t + \rho \sigma^Z \sigma^Y z \partial_{zy} f \dif t + \Delta f \dif D_t - \partial_z f \Delta Z_t.
\end{multline*}
{The pricing equation could be deduced by the $f$ dynamics in the same way discussed in Section \ref{sec:bk_pde}:}
\begin{subequations}\label{pde:generalUSD}
\begin{align}
 \partial_tv  & =
-    (r - \hat r) z  \partial_z v    -  a(b - y) \partial_y v    
- \frac 1 2 \left( \sigma^Z z\right)^2 \partial_{zz}v  \nonumber \\
&- \frac 1 2 \left( \sigma^{Y} \right)^2 \partial_{yy}v  - \rho \sigma^Z \sigma^Y z \partial_{zy} v  + e^{y}\ ( v - \gamma^Z z \partial_z v)\\
v(T, z, y) & = z.
\end{align}
\end{subequations}

\subsubsection{Inferring {D}efault {P}robability {D}evaluation {F}actor from the FX {R}ate {D}evaluation {F}actor} \label{sec:kImpl}
It is possible to link the FX rate devaluation factor introduced in \eqref{eqn:jumpFX} with a probability {rescaling} factor. 
This is done in the following
\begin{prop}[Default probabilities devaluation]\label{prop:gammaK}
Under the  hypotheses  of 
\begin{enumerate}[i)]
\item
small tenors:
\begin{equation}
T \rightarrow 0,
\end{equation}
\item
\added{independence between the Brownian motions driving  the FX and hazard rate processes:}
\begin{equation}
\rho =  0,
\end{equation}
\end{enumerate}
the ratio of the quanto-corrected and single-currency default probabilities can be approximated through
\begin{equation}\label{eqn:gamma2k}
\frac{1- \hat p_0(T) }{1- p_0(T) } \approx1 + \gamma^Z.
\end{equation}

\end{prop}
\begin{proof}
See Appendi{x} \ref{app:proof_gammaK}.
\end{proof}


\section{Results}\label{sec:results}

\subsection{{Numerical {M}ethods}}

{
In order to produce the results presented in this section, the PDE--system\eqref{pde:generalUSD} {has} been solved numerically, both for direct calculations of quanto--adjusted survival probabilities and for the calibration problems described later in the paper in Section \ref{sec:marketData}.
}

{
For this purpose, we implemented a finite--difference method belonging to the family of \emph{alternating--direction implicit} (ADI) schemes. 
The description of the scheme that has been  used can be found in \cite{ADI}. 
It must be noted that {the PDE system} \eqref{pde:generalUSD} {consists of} a pricing PDE and {of} a terminal condition.
In order to apply the chosen scheme to such PDE systems, we also have to specify boundary conditions. 
{For this purpose, we chose to use} neither Neumann nor Dirichlet conditions --- rather, the second derivative of the solution was set to zero on the boundaries.
}

\iftoggle{submission}{}{
\subsection{Monte Carlo {vs PDE for {D}efultable {B}ond {P}ricing {C}omparison}} 
We tested the PDE-based exponential OU implementation against a Monte Carlo one. To do so, we simulated a  hazard rate process {$\process{\lambda}$} given by $\lambda_t = e^Y_t$ with $\dif Y = a(b - Y_t)\dif t + \sigma\dif W_t$ where we used {a set of} parameters {such to produce a flat CDS par--spread term--structure at the level of 100 bps:}
\begin{align*}
a &= 0.08,&
b &= 3.7,&
\sigma &= 0.2,&
Y_0 &= -5.
\end{align*}
The values of the parameters above and the ones that we will be showing throughout the next sections are always expressed as annualized quantities.
We calculated numerically a domestic survival probability 
\begin{equation}\label{def:survProb}
p_0(T) = \E{0}{e^{-\int_0^T\lambda_s\dif s}}
\end{equation}
for $T=5Y$ and we reported the results in Figures \ref{fig:mctesta} and \ref{fig:mctestb}. The left-hand chart shows how, increasing the number of time steps used in MC simulation/PDE--grid, the MonteCarlo 95\%--confidence interval is moved up until it includes the PDE solution. As shown in Figure \ref{fig:mctesta}, this happened when the 5Y time horizon was sampled with at least 300 points. Figure \ref{fig:mctestb} shows how, fixing the number of time steps to 500, the 95\%--confidence interval is made smaller and smaller by increasing the number of MC paths but with the PDE--solution always lying inside of it.
\begin{figure}
\begin{floatrow}
\floatbox{figure}[.45\textwidth]
{\includegraphics[width = .45\textwidth]{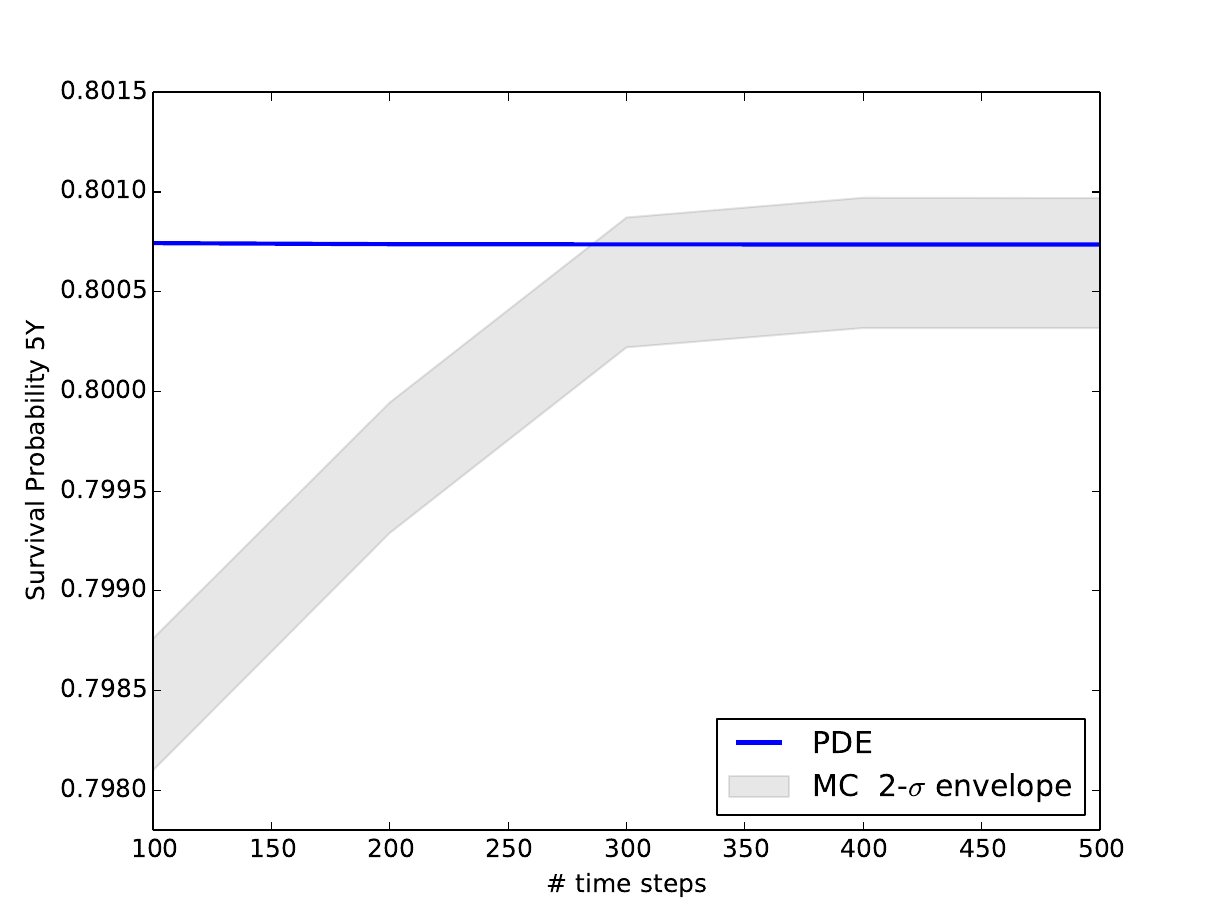}}
{%
  \caption{Comparison between PDE-based calculation of \eqref{def:survProb} and a MonteCarlo based one. The number of time steps used For both the methods is reported in the horizontal axis. The PDE-grid has been discretized with 200 points along the x-axis, while for MonteCarlo 100,000 paths have been used.}%
\label{fig:mctesta}
}
\floatbox{figure}[.45\textwidth]
{\includegraphics[width = .45\textwidth]{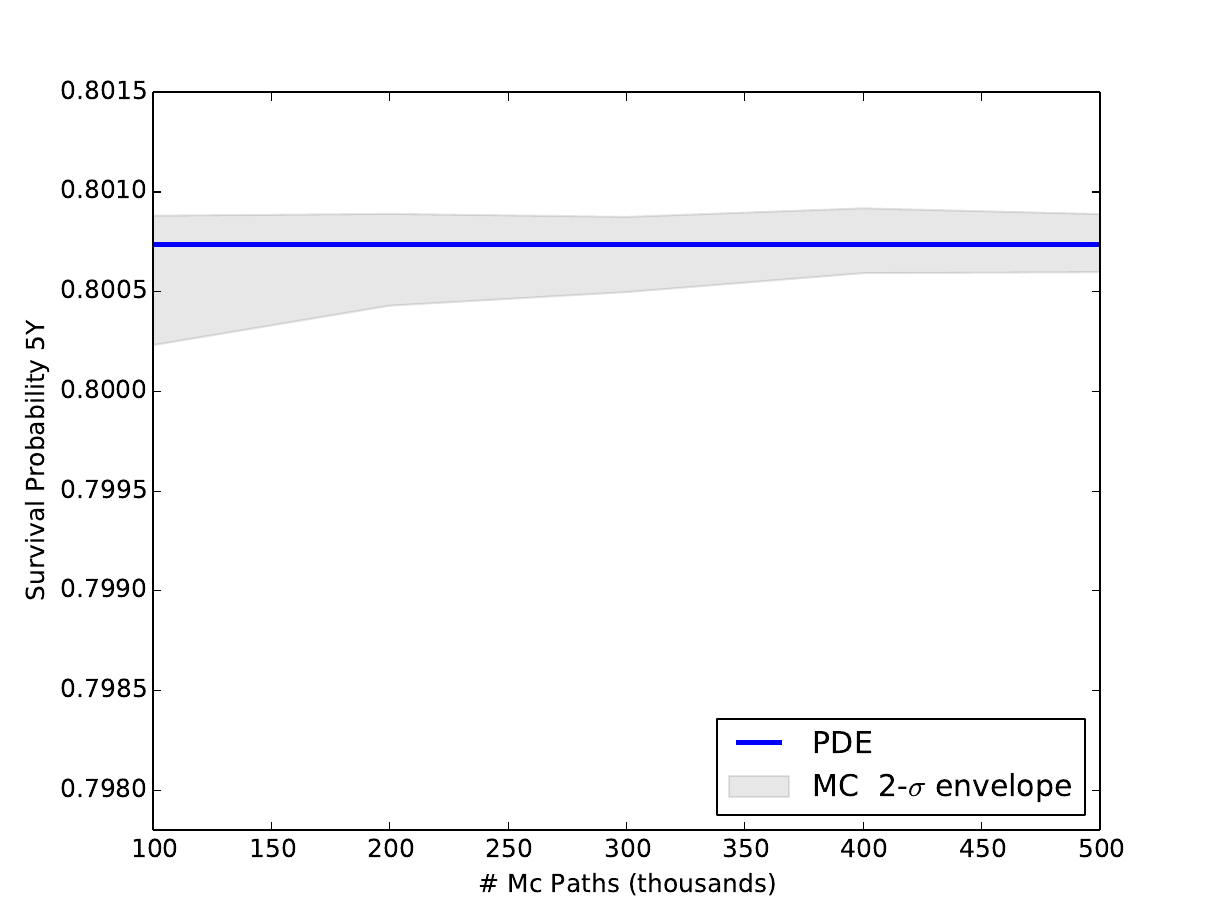}}
{%
  \caption{Same comparison of Figure \ref{fig:mctesta} where the number of time steps is fixed at 500 and the PDE-grid x-discretization to 200.}%
\label{fig:mctestb}
}
\end{floatrow}
\end{figure}
}

\subsection{{Quanto CDS {{P}ar--}{S}preads}  {P}arameters {D}ependence} \label{sec:parDependence}

In this section we show how the quanto-corrected {CDS par--spreads} are affected by changing the value of some of the parameters. Specifically

\begin{itemize}
\item
we show in Figure \ref{fig:yayb}  the dependence of {CDS} par--spread on the values of $\rho$ and $\gamma^Z$. 
\item
we show in Figure \ref{fig:sigmaFX} the dependence of {CDS} par--spreads on the value of $\sigma^{Z}$ for different values of $\sigma^Y$. For the ranges of values chosen, a stronger dependence is showed on $\sigma^Y$ than on $\sigma^{FX}$;
\item
we show in Figure \ref{fig:rhoSpread} the dependence of {CDS} par--spreads on the value of $\rho$ for different values of $\sigma^Y$. In particular, we show how the impact of correlation increases with $\sigma^Y$.
\end{itemize}


%
%
%

The parameter which affected the most the value of the spreads in this analysis is{, as one expects,} the devaluation rate, $\gamma^Z$ (see Figure \ref{fig:yayb}). 
For the {chosen} value of the parameters, a change in the instantaneous correlation from its extreme values, $-1$ and 1, can usually move the par spread of less than 10bps, while moving the devaluation rate to its extreme value, 1, can bring to zero the level of the par spread.

\begin{figure}
\begin{floatrow}
\floatbox{figure}[\textwidth]
{\includegraphics[width = \textwidth]{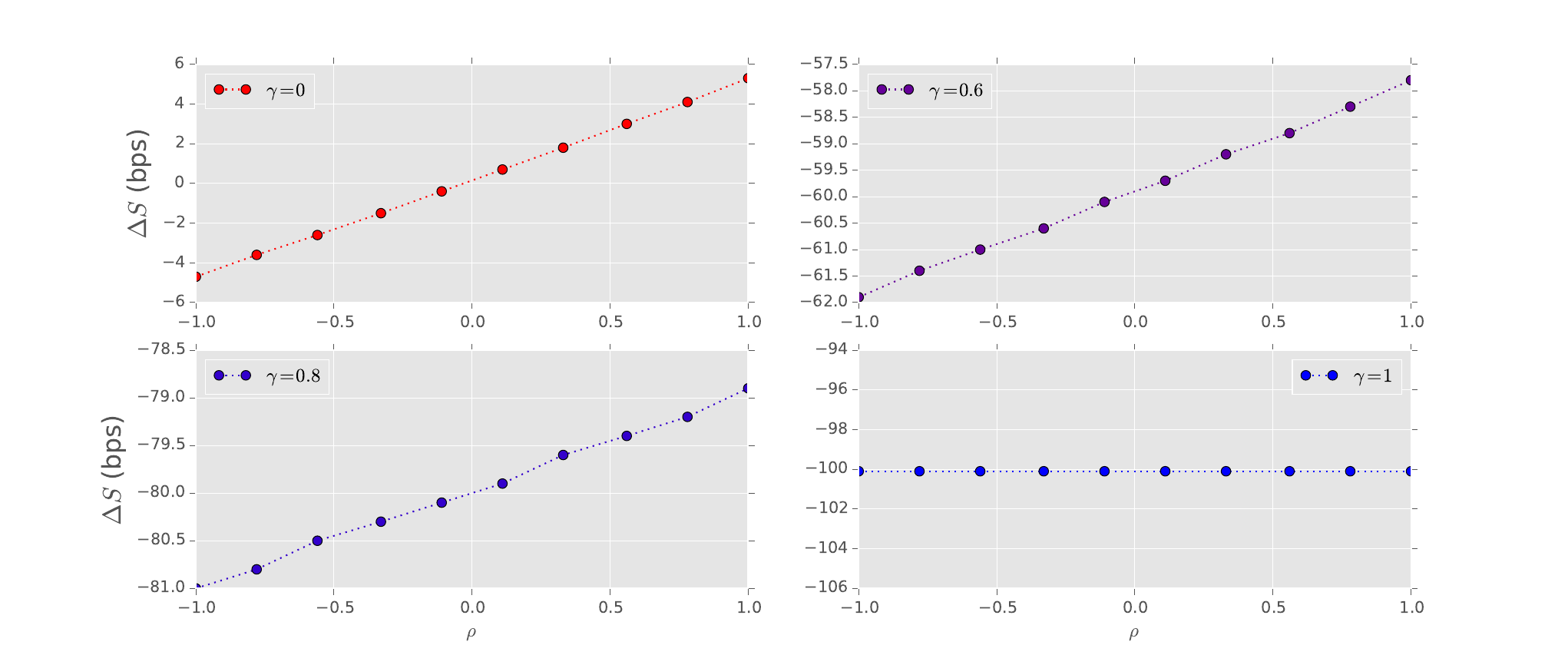}}
{%
  \caption{5Y CDS par--spread impact {vs} $\rho$ and $\gamma$. {The reference value for the par--spread is calculated using the parameters' values in Table \ref{tbl:yayb}.}}%
\label{fig:yayb}
}
\end{floatrow}
\end{figure}

\begin{table}
\centering
\begin{tabular}{ r@{.}l r@{.}l r@{.}l r@{.}l r@{.}l r@{.}l r@{.}l r@{.}l }
\toprule
\multicolumn{2}{ c }{$z$}&\multicolumn{2}{ c }{$\mu$}&\multicolumn{2}{ c }{$\sigma^Z$}&\multicolumn{2}{ c }{$a$}&\multicolumn{2}{ c }{$b$}&\multicolumn{2}{ c }{$y$}&\multicolumn{2}{ c }{$\sigma^Y$}&\multicolumn{2}{ c }{$T$}\\
\midrule
0&8&0&0&0&1&0&0001&-210&0&-4&089&0&2&5&0\\
\bottomrule
\end{tabular}
\caption{Parameters used to produce the par--spreads impact in Figure \ref{fig:yayb}}\label{tbl:yayb}
\end{table}

Figure \ref{fig:sigmaFX} shows that par--spreads' sensitivity to the volatility of the FX rate process is slightly weaker than the one to the log-hazard rate's volatility for the chosen ranges of parameters' values. {In our example, a} 5Y par--spread can change of around 10 bps with $\sigma^Z$ ranging from 1\% to 20\%,  while it can range up to 30 bps with $\sigma^Y$ going from 20\% to 70\% and with $\sigma^{Z}$ fixed at 20\%.

\begin{figure}
\begin{floatrow}
\floatbox{figure}[\textwidth]
{\includegraphics[width = \textwidth]{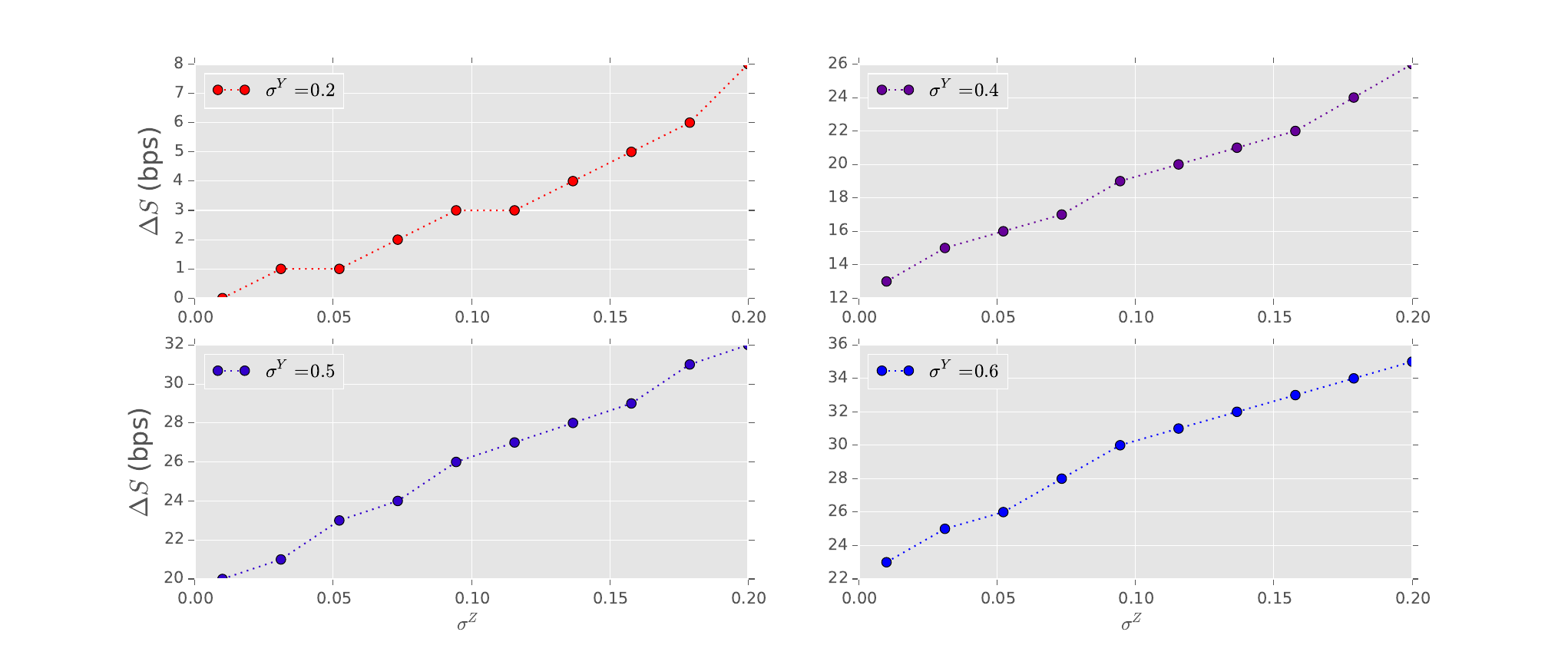}}
{%
  \caption{{5Y} CDS par--spread {impact vs $\sigma^Z$ and} $\sigma^Y$.
  {The reference value is produced using the parameters' values in Table \ref{tbl:sigmaFX}.}}%
\label{fig:sigmaFX}
}
\end{floatrow}
\end{figure}

\begin{table}
\centering
\begin{tabular}{ r@{.}l r@{.}l r@{.}l r@{.}l r@{.}l r@{.}l r@{.}l }
\toprule
\multicolumn{2}{ c }{$z$}&\multicolumn{2}{ c }{$\mu$}&\multicolumn{2}{ c }{$\rho$}&\multicolumn{2}{ c }{$a$}&\multicolumn{2}{ c }{$b$}&\multicolumn{2}{ c }{$y$}&\multicolumn{2}{ c }{$T$}\\
\midrule
0&8&0&0&0&5&0&0001&-210&0&-4&089&5&0\\
\bottomrule
\end{tabular}
  \caption{Parameters used to produce the results shown in Figure \ref{fig:sigmaFX}.}%
  \label{tbl:sigmaFX}
\end{table}



In Figure \ref{fig:rhoSpread} we show the sensitivity of par--spreads to the value of diffusive correlation $\rho$.  The dependence of par spreads on the correlation is extremely weak for values of $\sigma^Y$ in the range of 20\%. Around this level of log-hazard rate volatility, the maximum change that correlation can produce on the quanto-par spreads is 10 bps. 
From Figure \ref{fig:rhoSpread}, a {more realistic} value of $\sigma^{Y}$ of {60}\%  is required to observe an impact of around {30} bps on the 5Y par--spread when changing the correlation from $-1$ to 1, showing the limits of a purely diffusive correlation model in explaining large differences between domestic and quanto-corrected {CDS par-}spreads.

\begin{figure}
\begin{floatrow}
\floatbox{figure}[\textwidth]
{\includegraphics[width = \textwidth]{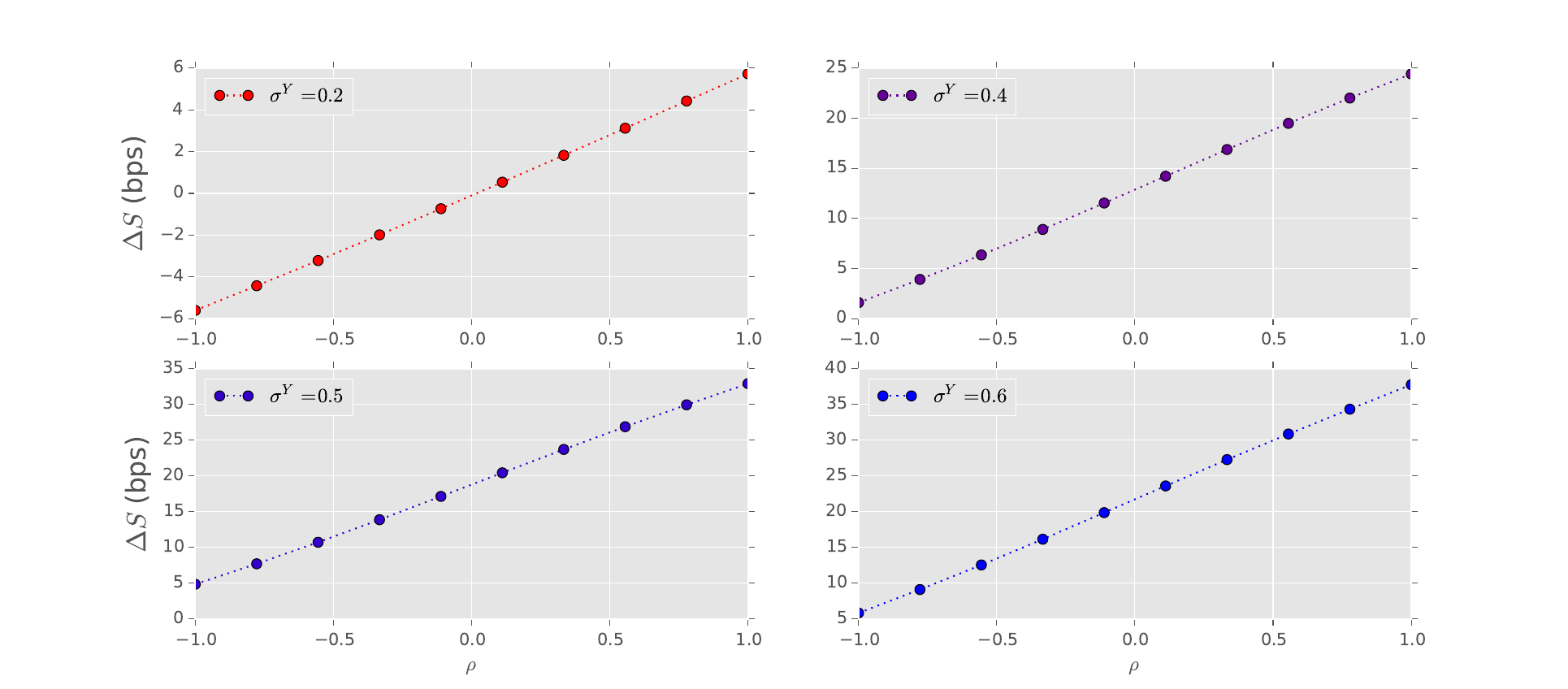}}
{%
  \caption{5 years par--spread {impact vs $\rho$ and} $\sigma^Y$. {The reference value is produced using the parameters' values in Table \ref{tbl:rhoSpread}.}}%
\label{fig:rhoSpread}
}
\end{floatrow}
\end{figure}

\begin{table}
\centering
\begin{tabular}{ r@{.}l r@{.}l r@{.}l r@{.}l r@{.}l r@{.}l r@{.}l }
\toprule
\multicolumn{2}{ c }{$z$}&\multicolumn{2}{ c }{$\mu$}&\multicolumn{2}{ c }{$\sigma^Z$}&\multicolumn{2}{ c }{$a$}&\multicolumn{2}{ c }{$b$}&\multicolumn{2}{ c }{$y$}&\multicolumn{2}{ c }{$T$}\\
\midrule
0&8&0&0&0&1&0&0001&204&0&-4&089&5&0\\
\bottomrule
\end{tabular}
  \caption{Parameters used to produce the results shown in Figure \ref{fig:rhoSpread}.}%
\label{tbl:rhoSpread}
\end{table}

There are circumstances where the  basis between par--spreads  of CDSs in different currencies can be sensibly higher than these values. In those cases, a purely diffusive model for the hazard rate is not sufficient to explain the observed basis and an approach where dependence is induced by devaluation jumps is required.  As an example of an historical occurrence of such a wide basis, we show in Section \ref{sec:backtest} results of model calibrations to  the time series of  par--spreads for EUR-denominated and USD-denominated 5Y CDSs on the Italian Republic. 

In the different context of impact of dependence on CDS credit valuation adjustments, even under collateralization, Brigo et al \cite{brigo3,brigo4, BMP} show that a copula function on the jump to default exponential thresholds may be necessary to obtain sizable effects when looking at credit--credit dependence, pure diffusive correlation not being enough.

\subsection{Test on {the {I}mpact of {T}enor and {C}redit {W}orthiness on the {Q}uanto {C}orrection}}\label{sec:otherTest}

To test the relation given in Eq \eqref{eqn:gamma2k}, we set the diffusive correlation to zero and we chose the following set of log-hazard rate parameters:
\begin{align*}
a &= 1.00e-004,&
b &=-210.45&
\sigma &= 0.2,&
\end{align*}
whereas we have produced low spread scenarios and high spread scenarios by choosing two different values for $Y_0$, the first one, $y^L =-4.089$, such that the resulting CDS par spread term structure is flat at around 100 bps, and the second one, $y^H =-2.089$ such to produce a flat CDS par spread term structure with a value of around 740 bps.

Figures \ref{fig:kGammaLow} and \ref{fig:kGammaHigh} show, in line with the nature of the approximation \eqref{eqn:gamma2k}, that the approximation is less accurate for higher maturities, as evident in both charts by comparing blue lines (short maturities) with red lines (long maturities){. It is also less accurate} and for higher values of CDS spreads, as highlighted by the comparison between Figure \ref{fig:kGammaHigh} (high spreads) and Figure \ref{fig:kGammaLow} (low spreads).

\begin{figure}
\begin{floatrow}
\floatbox{figure}[.5\textwidth]
{\includegraphics[width = .5\textwidth]{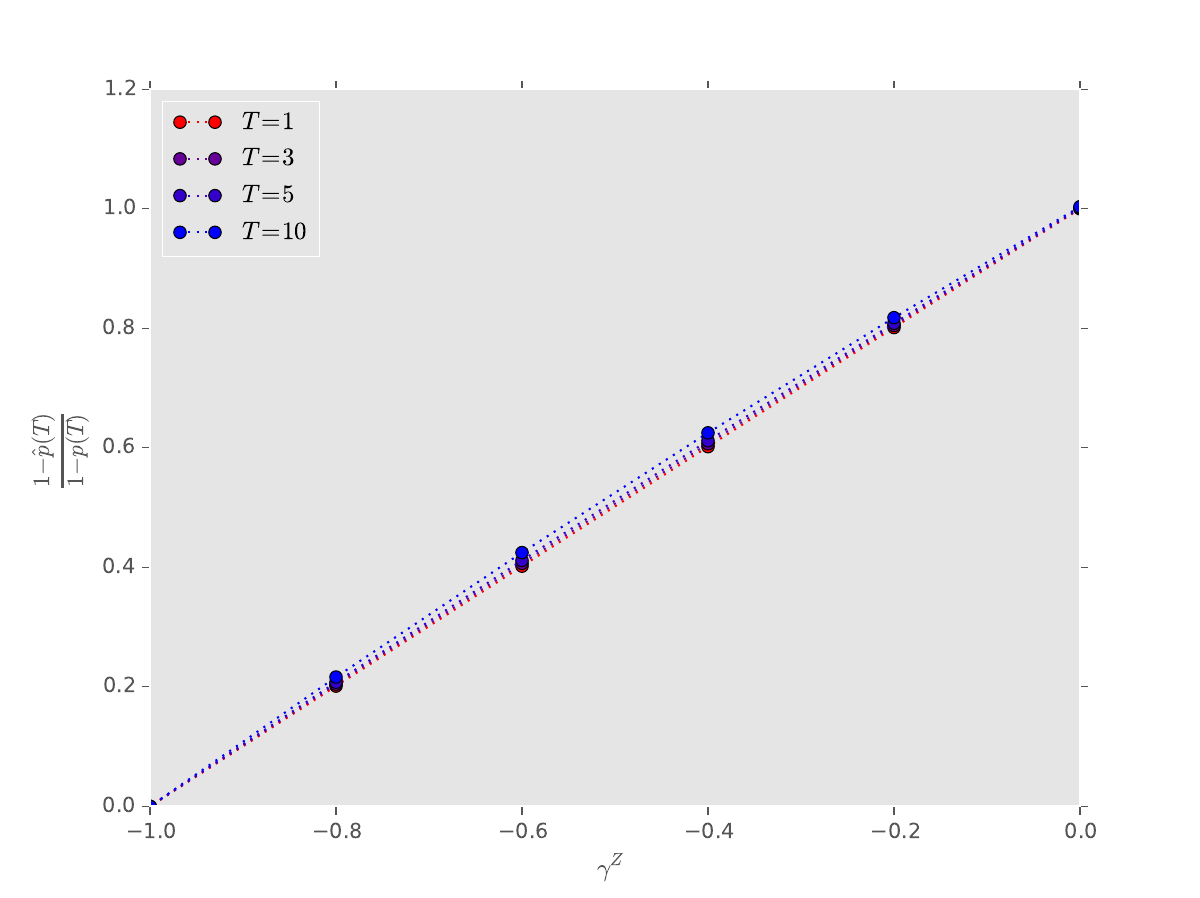}}
{%
  \caption{Comparison of curves $\sfrac{1-\hat P_0(T)}{1-P_0(T)}$ for different maturities in a low spreads scenario, $Y_0 = y^L$. }%
\label{fig:kGammaLow}
}
\floatbox{figure}[.5\textwidth]
{\includegraphics[width = .5\textwidth]{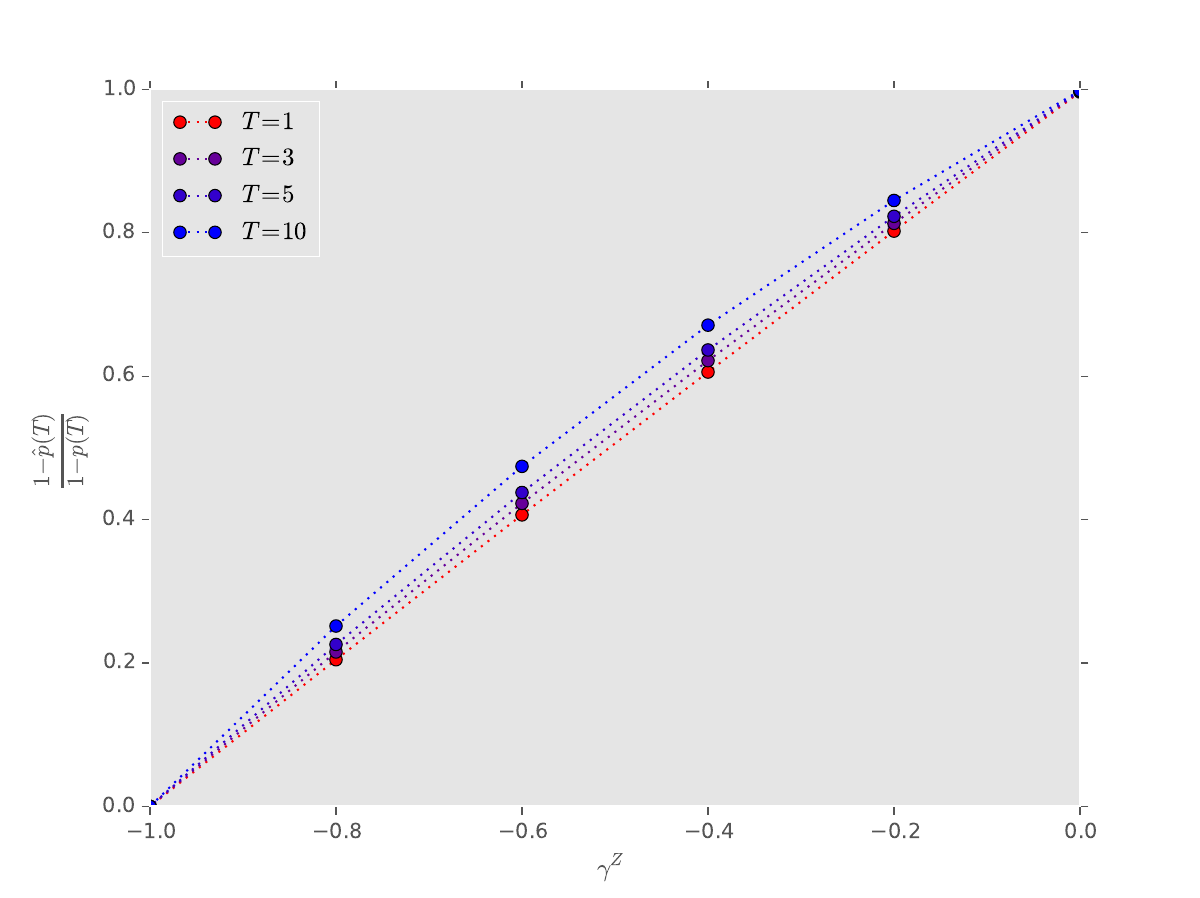}}
{%
  \caption{Comparison of $\sfrac{1-\hat P_0(T)}{1-P_0(T)}$ for different maturities  in a high spreads scenario, $Y_0 = y^H$.}%
\label{fig:kGammaHigh}
}
\end{floatrow}
\end{figure}

\subsection{Correlation {I}mpact on the {S}hort {T}erm {V}ersus {L}ong {T}erm}

We checked numerically the robustness of the theoretical relation between survival probabilities and $\gamma^Z$ that was shown in Eq \eqref{eqn:gamma2k}. We calculated the ratio between the local and the quanto-corrected survival probability returned by the exponential OU model for different maturities and for different values of $\rho$. Furthermore, we express this value as a function of $\gamma^Z$, we call it $\hat q(\gamma)$, and we check how this value {is} affected by changes in $\gamma$. We then compare $\hat q$ with the limit-case value provided by Eq \eqref{eqn:gamma2k}
\begin{equation}
q(\gamma) := 1 + \gamma.
\end{equation}

The results, in the form of a percentage difference $\sfrac{q}{\hat q} - 1$,  are reported in Table \ref{tbl:kImpl}.

\begin{table}
\centering\small
\begin{tabular}{ r@{.}l r@{.}l r@{.}l r@{.}l c r@{.}l r@{.}l r@{.}l c r@{.}l r@{.}l r@{.}l }
\toprule
\multicolumn{2}{ c }{}&\multicolumn{6}{ c }{ $T = 1$ }& &\multicolumn{6}{ c }{$T = 4$}& &\multicolumn{6}{ c }{$T = 10$}\\
\cmidrule{3-8} \cmidrule{10-15} \cmidrule{17-22}
\multicolumn{2}{ c }{$\gamma$}&\multicolumn{2}{ c }{$ \rho= -0.9 $}&\multicolumn{2}{ c }{$ \rho= 0 $}&\multicolumn{2}{ c }{$ \rho= 0.9 $}& &\multicolumn{2}{ c }{$ \rho= -0.9 $}&\multicolumn{2}{ c }{$ \rho= 0 $}&\multicolumn{2}{ c }{$ \rho= 0.9 $}& &\multicolumn{2}{ c }{$ \rho= -0.9 $}&\multicolumn{2}{ c }{$ \rho= 0 $}&\multicolumn{2}{ c }{$ \rho= 0.9 $}\\
\midrule
-0&99&0&24 \%&0&08 \%&$-0$&07 \%& &$-2$&60 \%&$-3$&33 \%&$-4$&05 \%& &$-31$&74 \%&$-34$&93 \%&-35&58 \%\\
-0&50&0&43 \%&0&28 \%&0&12 \%& &$-0$&15 \%&$-0$&89 \%&$-1$&62 \%& &$-30$&33 \%&$-26$&68 \%&$-25$&72 \%\\
-0&25&0&53 \%&0&37 \%&0&22 \%& &1&11 \%&0&37 \%&$-0$&36 \%& &$-15$&84 \%&$-15$&01 \%&$-14$&87 \%\\
0&00&0&63 \%&0&47 \%&0&31 \%& &2&38 \%&1&65 \%&0&91 \%& &-1&16 \%&$-1$&26 \%&$-1$&37 \%\\
0&25&0&73 \%&0&57 \%&0&41 \%& &3&67 \%&2&93 \%&2&20 \%& &13&47 \%&13&71 \%&14&60 \%\\
0&50&0&83 \%&0&67 \%&0&51 \%& &4&96 \%&4&23 \%&3&50 \%& &28&04 \%&29&82 \%&35&40 \%\\
\bottomrule
\end{tabular}
\caption{Deviation from the relation in \eqref{eqn:gamma2k} expressed as percentage error between $1+\gamma$ and $\hat q(\gamma) $.}\label{tbl:kImpl}
\end{table}

They show, as hinted in \cite{FxBofa}, that correlation has a smaller impact on short term survival probabilities: moving the correlation between the values of $-0.9$ and $0.9$ has an absolute impact of 0.3\% on our results for 1Y survival probabilities, whereas the impact for 4Y survival probabilities is 1.45\% and for 10Y survival probabilities is almost 4\%. It has to be noted that for 10Y survival probabilities, $\hat q$ doesn't provide a good approximation of $q$  not even in case of null correlation. This last fact is in line with the discussion carried out in Section \ref{sec:kImpl}, as in this case the hypothes{e}s under which the approximation was deduced are not valid.

\subsection{Model {C}alibration to {M}arket {D}ata for  2011--2013}\label{sec:backtest}
In this section we present the results of the calibration of {the} model described in Section \ref{sec:pEqJumps}, {where pricing currency and liquid currency coincide and are USD, and where we considered two contractual currencies, EUR and USD}. 

We used the observed CDSs spreads on Italy, both the USD-denominated ones and the EUR-denominated ones, to calibrate the model parameters. In principle, also single-name CDS swaptions could be used in this calibration process {(see \cite{brigo2})}, but, given the lack of liquidity on this instrument{, we preferred proxying them with the at-the-money implied volatilities quoted for options on iTraxx Main}.
\subsubsection{Market {D}ata {D}escription}\label{sec:marketData}
We calibrated the model to the market data for the three years using the time range 2011-2013. 
Let $\mathcal T = \{t_0\dots, t_N\}$ denote the dates in this sample period. We made the following assumptions on the market data:

\begin{enumerate}[i)]
\item
we consider the CDS par--spreads on Republic of Italy with  5 years and 10 years {tenor}, both in USD and in EUR;
\item
we use the same short rate for domestic and foreign currency
\begin{equation}
r(t_i) = \hat r(t_i) = r,\quad t_i \in \mathcal T,
\end{equation}
\item
on every $t_i\in \mathcal T$ we assign the {value of the} at--the--money  Black--volatility {from an option with 6 months expiry} to $\sigma^Z$;
\item
we {keep} the speed of mean reversion $a$ of $\process Y$ flat at the level $\num{0.0001}$;
\item
on every $t_i\in \mathcal T$ we calibrated $\sigma^Y$  to the at-the-money option Black volatility for expiry one month.
\end{enumerate}

Denoting by $p^Y:=(b, y_0)$ the parameters to be calibrated for $\process Y$ that are needed in single currency CDS pricing,  and by $p:= (b, y_0, \rho, \gamma)$ the set of parameters needed to price a quanto CDS, we {adopted} the following procedure to calibrate the model  in Eq \eqref{eqn:yEUR}--\eqref{eqn:rhoEUR}:
\begin{enumerate}[i)]
\item
first we calibrated $p^Y$ to the USD-denominated par--spread for the given date. We kept the parameters $a$ and $\sigma^Y$ fixed at a level of  \num{0.0001}  and 50\% respectively;
\item
we calibrated $\sigma^Y$ to the CDS index option, keeping the $p^Y$ at the level calibrated at the previous step;
\item
we used the calibrated value of $p^Y$ as a starting point in the iterative routine carried out to calibrate the set of model parameters $p$ to both the EUR-denominated and the USD-denominated CDSs. The starting guess point to calibrate $p$ can be written in terms of the calibrated point $p^Y$ as $p_0 = (p^Y_1, p^Y_2, \gamma_0, \rho_0)$, where $\gamma_0$ and $\rho_0$ are the guess values for $\gamma$ and $\rho$. We kept $\sigma^Y$ fixed at the level calibrated at the previous step.
\end{enumerate}

\subsubsection{Results}\label{sec:basktextResults}
In this section we show the results of the {repeated daily} calibrations to the 3 years of data contained in $\mathcal T$. The calibrated $\gamma$ and $\rho$ are showed in Figure \ref{fig:gammaRhoBasis} together with the relevant market data used in calibration, EUR-denominated and USD-denominated CDS par spreads  for 5 years maturities, $\EUR{\Y{S}{5}}$ and $\USD{\Y{S}{5}}$, and for 10 years  maturities, $\EUR{\Y{S}{10}}$ and $\USD{\Y{S}{10}}$ .

\begin{figure}
\centering
\includegraphics[width = \textwidth]{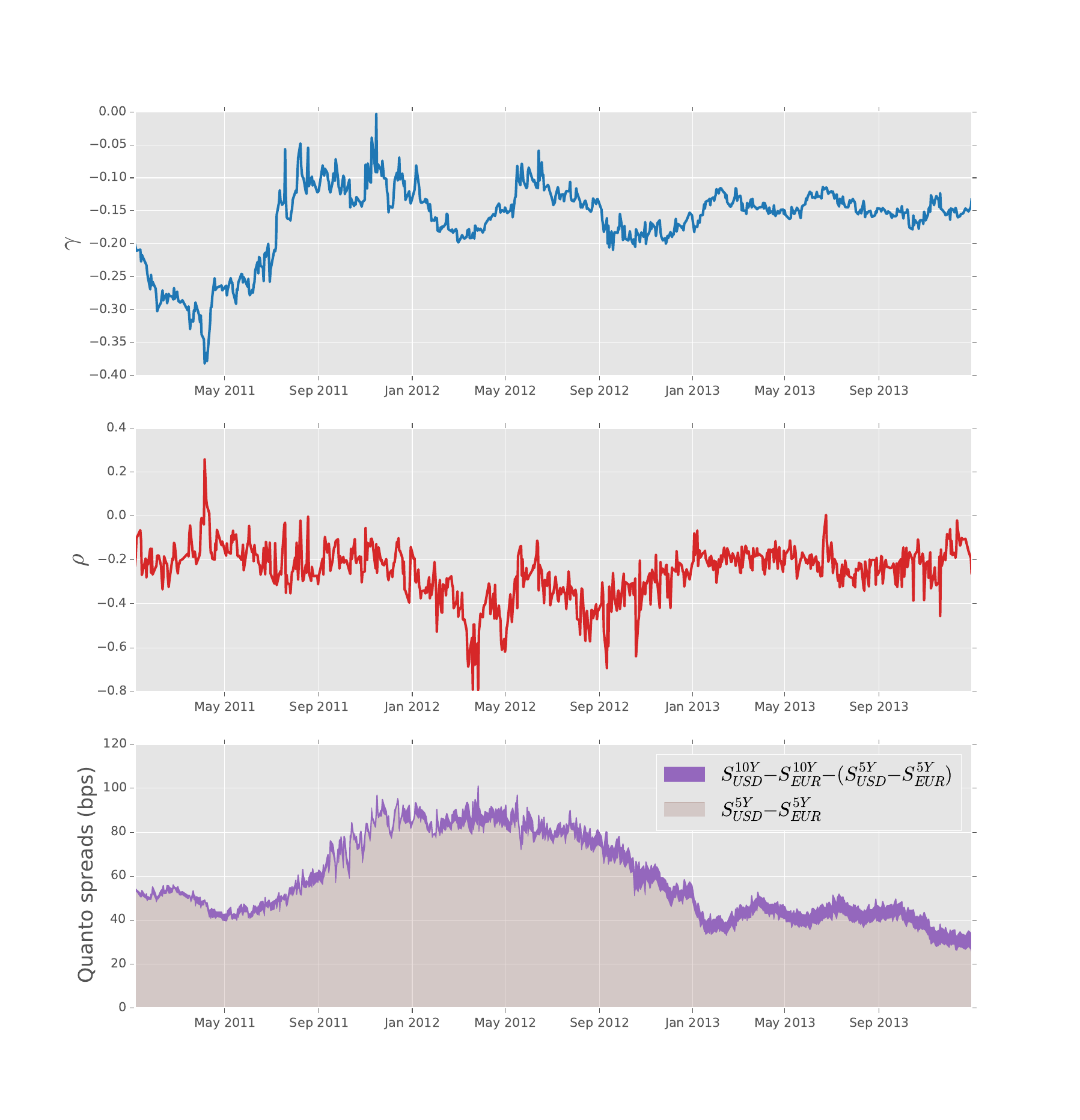}
 \caption{{The t}op chart shows the calibrated $\gamma$ throught $\mathcal T$. {The m}iddle chart shows the calibrated $\rho$ thorught $\mathcal T$. {The b}ottom chart shows the time series of {corresponding quanto }CDS par--spreads.}%
\label{fig:gammaRhoBasis}
\end{figure}

The aim of this section is to interpret the calibrated parameters in terms of market data. To do so, we{ will} be relying on the theoretical results from the previous section. 

\paragraph{Interpretation of the {D}evaluation {F}actor $\gamma^Z$}

For the devaluation rate, $\gamma^Z$, we exploited the results from Section \ref{sec:kImpl}, and we used the  relative basis spreads as an approximation
\begin{equation}\label{eqn:approxGammaZ}
\gamma^Z \approx \frac{\EUR{S} - \USD{S}}{\USD{S}}.
\end{equation}
As shown in Proposition \ref{prop:gammaK}, the simplified relation between $\gamma^Z$ and the ratio of the quanto and non--quanto corrected default probabilities is true for small values of the quantity $\int_t^T \lambda_s \dif s$ so, {since we could not} control the credit quality in backtest, we relied on the time--to--maturity $T-t$ to achieve a good approximation. However, due to liquidity reasons, we used CDS par--spreads with   5 years and 10 years {tenor}, and these  maturity values {can be} too large. Therefore we used model--implied par--spreads for this test; in this way we have been able to use also short maturities, like 1 year, that are usually not very liquid in the market.

The comparison between $\gamma^Z$ and its market-data approximation is showed in Figure \ref{fig:gammaScatter}. The left-hand chart, where 1Y-spreads have been used to build the relative basis spread, shows a surprisingly good agreement between the two variables. The same agreement does not hold for the right-hand chart, where 5Y-spreads were instead used. This is in line with the result of Proposition \ref{prop:gammaK}, that was derived under a limit hypothesis of short maturities.

{It is worth highlighting that the approximation provided by Eq \eqref{eqn:approxGammaZ} would be an exact relation between $\gamma^Z$, $\EUR S$, and $\USD S$ for contracts for which it is possible to approximate the stream of the premium leg's quarterly-spaced cash-flows  with a continuously compounded stream of payments and in a setting where either the hazard rate was modeled as a deterministic function of time and where the CDS par--spread term structure was flat or in a setting where the hazard rate was modeled as a constant.}

\begin{figure}
\subfloat[Relative basis spread for 1Y maturity CDSs.]
{\includegraphics[width = .49\textwidth]{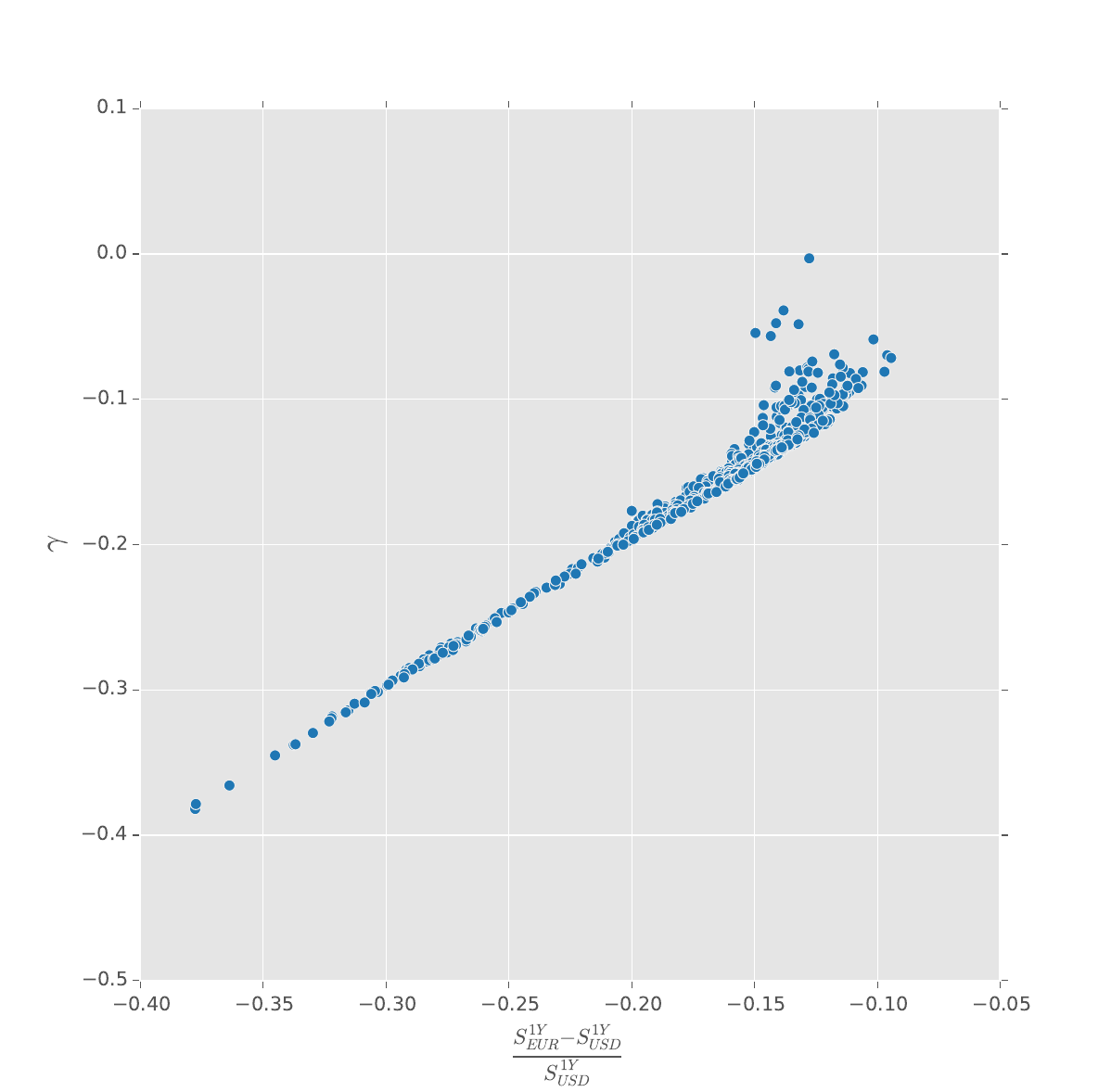}}
{%
\label{fig:gamma1YScatter}
}
\subfloat[Relative basis spread for 5Y maturity CDSs.]
{\includegraphics[width = .49\textwidth]{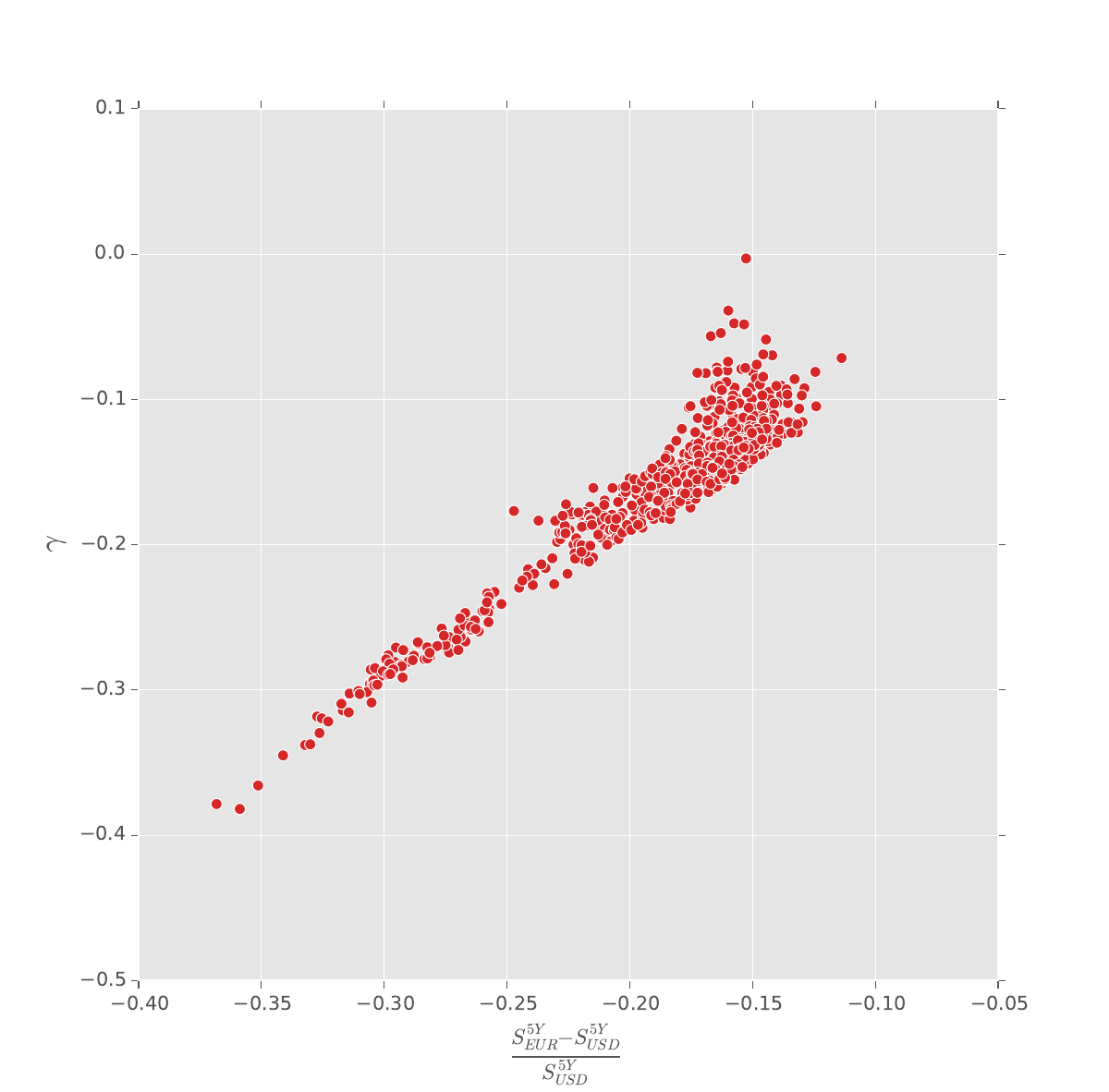}}
{%
\label{fig:gamma5YScatter}
}
\caption{Scatter plot comparing the calibrated $\gamma^Z$ in ordinates with a relative basis spread in abscissas. As discussed in Section \ref{sec:basktextResults}, the chart shows that the relative basis spread is a good estimate of the devaluation rate if the spreads have short maturities}
\label{fig:gammaScatter}
\end{figure}

\paragraph{Interpretation of the {I}nstantaneous {C}orrelation {P}arameter $\rho$}

In order to provide a similar assessment on the parameter $\rho$, we relied on some heuristic results derived in \cite{jpmqcds}. In that technical report, a simplified pricing formula based on cost of hedging arguments is presented for quanto CDS .
Their result can be written in terms of the variable defined by our framework as
\begin{equation}
\frac{\EUR{S}(T) - \USD{S}(T)}{\USD{S}(T)} \approx {\gamma^Z} + \sigma^Y \sigma^Z \rho \rpv(T),
\end{equation}
{where $\rpv(t)$ is the risky annuity of a CDS with tenor $t$ years.}
We applied the formula above to two tenor points $T_1$ and $T_2$ obtaining two equations, one for each tenor. In order to  test the values of $\rho$ that we obtained in calibration, we worked out a single equation as a difference between the equations for the two tenor points:
\begin{equation}\label{eqn:jpm}
\frac{\EUR{S}(T_2) - \USD{S}(T_2)}{\USD{S}(T_2)}  - \frac{\EUR{S}(T_1) - \USD{S}(T_1)}{\USD{S}(T_1)} \approx  \sigma^Y \sigma^{Z} \rho \left(\rpv(T_2) - \rpv(T_1)\right).
\end{equation}
Specifically, we chose $T_1 = 1$, $T_2 = 10$ and we used the model--implied values of $\EUR S(T_1)$, $\EUR S(T_2)$, $\USD S(T_1)$, $\USD S(T_2)$, $\rpv(T_1)$ and $\rpv(T_2)$. We further used  the values $\sigma^{Z}$ coming from the market while the values of $\sigma^Y$ and $\rho$ are the ones obtained in calibration ad discussed in Section \ref{sec:marketData}.
The results are presented in Figure \ref{fig:rhoScatter} and they show {a scatterplot} of the proposed relation between model parameters and market data. 
The data are reported for the whole time--range 2011--2013 in Figure \ref{fig:rhoScatter3Y}, {while Figure \ref{fig:rhoScatterYoY} contains the} year--by--year {plot}.
Due to the empirical nature of the Eq \eqref{eqn:jpm}, we didn't expect to find an exact relation between $\rho$ and other model parameters and market data. Nonetheless, a clear pattern is exhibited and this gives some confidence that such relation can be used to produce at least rough approximations for $\rho$ by using observable market data. 

\begin{figure}
\subfloat[Scatter plot for the time--range 2011--2013]
{\includegraphics[width = .49\textwidth]{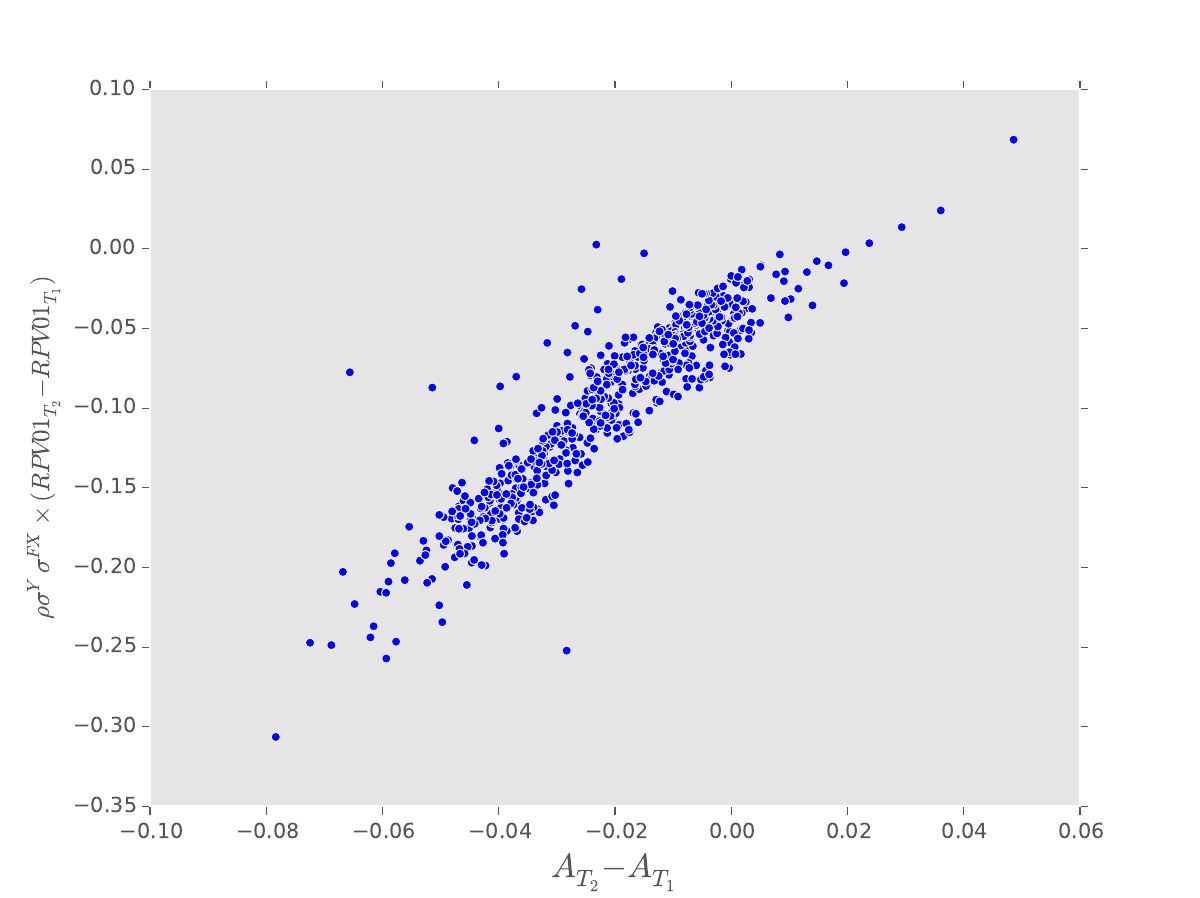}\label{fig:rhoScatter3Y}}
{%
}
\subfloat[Split by year]
{\includegraphics[width = .49\textwidth]{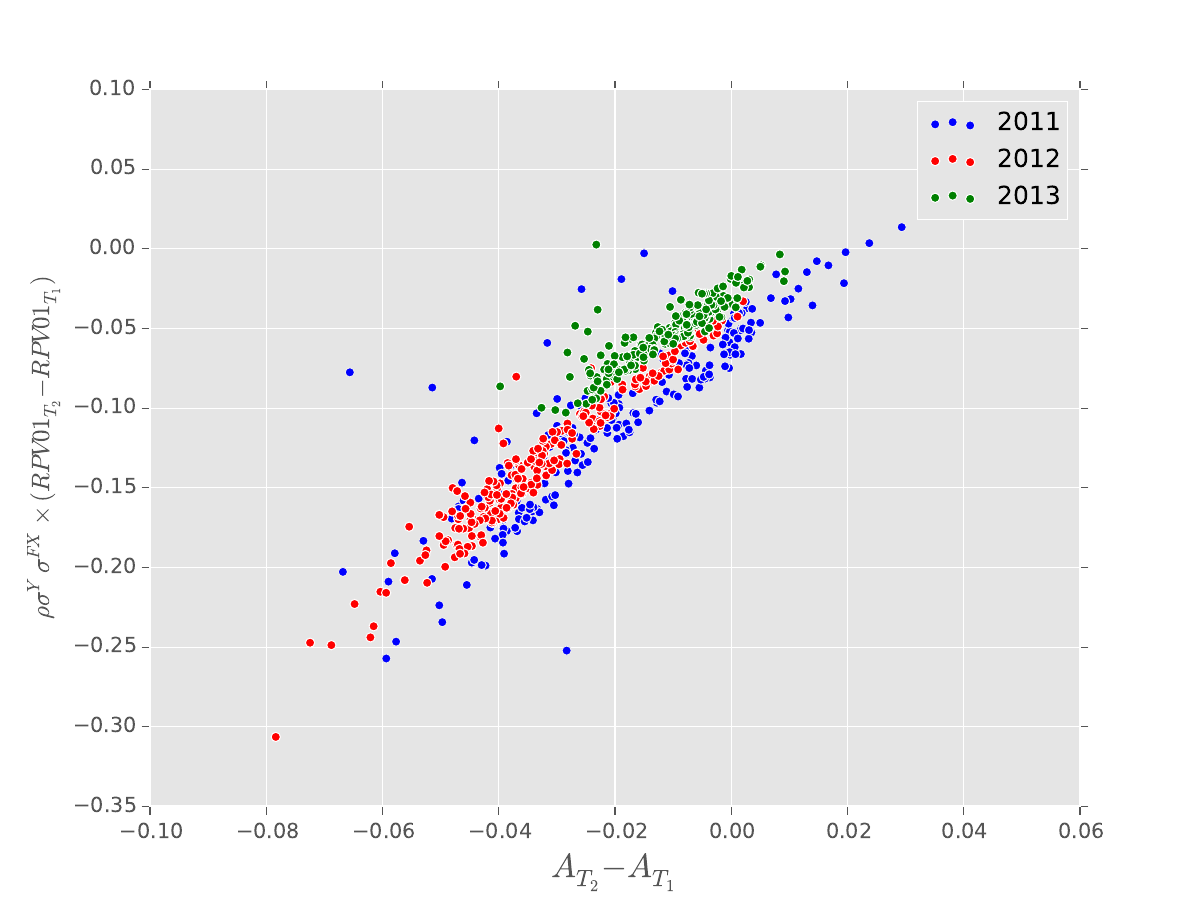}\label{fig:rhoScatterYoY}}
{%
}
\caption{Scatter plot comparing the product of the calibrated $\rho$, $\sigma^Y$,$\sigma^Z$ and the model--implied difference between risky annuities in ordinates with a difference of relative basis spread in abscissas (we used $A_T:= \frac{\EUR{S}(T) - \USD{S}(T)}{\USD{S}(T)}$). }
\label{fig:rhoScatter}
\end{figure}

\paragraph{{M}odel--{I}mplied vs {H}istorical {C}orrelation}
In Figure \ref{fig:histCorr50} we reported a comparison between the correlation parameter we obtained in calibration, $\rho$, and a historical estimator of correlation between {daily log-returns of CDS par--spreads for one--year tenor contracts and daily log--returns of the FX spot rate.
For assets where the market correctly prices gamma and cross--gamma risks, the basis between implied and realised covariance terms can be actually traded. This happens, for example, for implied and historical volatilities on  equity indices.}

{In times where the values of implied and realised covariance terms diverge, the effect of such trading strategies is usually to bring them closer. We interpret the lack of evident convergence between implied and realised correlation in the chart in Figure \ref{fig:histCorr50} as a signal of the lack of an efficient market for this correlation risk. }
\begin{figure}
{\includegraphics[width = \textwidth]{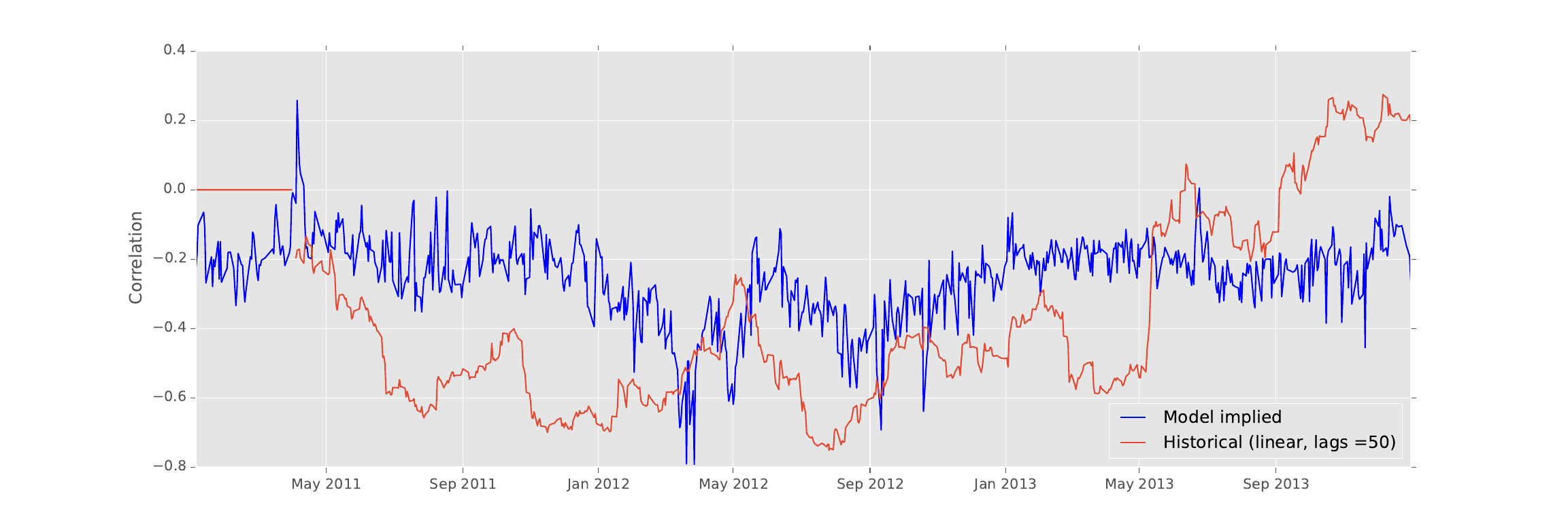}}
{%
  \caption{Implied and historical correlation  between EURUSD FX rate and Italy's CDS spread.}%
\label{fig:histCorr50}
}
\end{figure}

{The fact that the implied correlation is generally smaller in absolute value than the realised one is consistent with our modeling choices and with the estimator used to calculate the realised correlation. 
The historical correlation has been estimated {on a 50 days time--window} using log--returns of the FX rate and this would neglect the impact of the jump term on its instantaneous volatility. Such an underestimation of the instantaneous volatility of the jump--diffusion process used in our modelling approach would result in an overestimation of the correlation with the credit component.}

\section{Conclusions and Further Work}

We analysed default--driven FX devaluation jumps as a modelling mechanism. 
These can be used to explain the  basis in credit default swaps offering protection on the same entity but in different currencies. 
We studied the case of Italy and EUR vs USD protection in particular. 
We found that the jump mechanism allows one to explain the size of the basis, whereas pure shock correlation between FX rates and credit spread is not sufficient. 
Further applications we may consider in future work include wrong--way risk modelling in credit valuation adjustment (CVA) applications.

\begin{appendices}
\iftoggle{submission}{}{
\section{Some  results for semimartingales}
\subsection{Ito formula}\label{app:ito}

The usual Ito formula for continuous semimartingales and a $\mathcal C^2$ function $f$ can be written as
\begin{equation}\label{eqn:itoContin}
f(X_T) - f(X_t) = \int_t^T f'(X_{s-}) \dif X_s + \frac 1 2 \int_t^T f''(X_{s-}) \dif\ [X]_s
\end{equation}
where $\process{[X]}$ is the quadratic variation of the process $\process X$.
If we wish to include discontinuous processes, we need to compensate the previous formula, accounting for jumps in the left--hand side and in the right--hand side:
\begin{align*}
\mbox{left--hand side jumps:} & \sum_{t\leq s \leq T} (f(X_s)-f(X_{s-}))\\
\mbox{right--hand side jumps:} & \sum_{t\leq s \leq T} f' \Delta X_s + \frac 1 2 \sum f'' \Delta X_s^2
\end{align*}
where we used $X_{t-} := \lim_{s\rightarrow t}X_s$ and
\begin{equation*}
\Delta X_t := X_t - X_{t-}.
\end{equation*}
Introducing the jump terms, \eqref{eqn:itoContin} becomes
\begin{multline}\label{eqn:ito1}
f(X_T) - f(X_t) = \int_t^T f'(X_{s-}) \dif X_s + \frac 1 2 \int_t^T f''(X_{s-}) \dif\ [X]_s\\
 +\sum_{t\leq s \leq T} \left(f(X_s)-f(X_{s-}) - f'(X_{s-}) \Delta X_s - \frac 1 2 f''(X_{s-}) \Delta X_s^2 \right)
\end{multline}

The quadratic variation of $\process X$ is a FV process, hence it makes sense (because the sum in the next equation converges) to define its \emph{continuous part} as 
\begin{equation}
[X]^c_t := [X]_t - \sum_{s\leq t}\Delta X_s^2, \quad t > 0,
\end{equation}
so that \eqref{eqn:ito1} simplifies into
\begin{equation}
f(X_T) - f(X_t) = \int_t^T f'(X_{s-}) \dif X_s + \frac 1 2 \int_t^T f''(X_{s-}) \dif\ [X]^c_s
+ \sum_{t\leq s \leq T} \left(f(X_s)-f(X_{s-}) - f'(X_{s-}) \Delta X_s \right)
\end{equation}

We remind the interested reader to \cite{Protter} for a proof of the results above.
\begin{remark}[Processes with FV jumps]
If, as it is the case for the rest of this work, we consider semimartingales whose jumps component is a finite--variation process, also the sum
\begin{equation}
\sum_{s\leq t}f' \Delta_sX <\infty
\end{equation}
\end{remark}

\subsection{ Girsanov Theorem}\label{app:girsanov}
Let us consider given a filtered probability space $(\Omega, \mathcal F, \mathbb P, (\mathcal F_t))$ satisfying the usual hypothesis. In the following, we will always consider stochastic processes adapted to this space, or to a variation of it, where we will change the probability measure if needed.
Suppose $\process L$ is the Radon--Nikodym derivative that is used to change from a certain probability measure $\mathbb M$ to an equivalent measure $\mathbb N$. By the Change of Numeraire Theorem, the price at time $t$, say $V_t$, of a derivative paying off $\phi_T$ at time $T > t$, can be expressed equivalently using either one of the two measures as
\begin{equation}\label{eqn:chNumer}
V_t = M_t\mathbb E^{\mathbb M}\left[\left. \frac{\phi_T}{M_T}\right | \mathcal F_t\right] = N_t\mathbb E^{\mathbb N}\left[ \left. \frac{\phi_T}{N_T}\right | \mathcal F_t\right],
\end{equation}
where we used $\mathbb E^\mathbb X$ to denote the expected value with respect to the probability measure $\mathbb X$.
From Eq \eqref{eqn:chNumer}, the Radon--Nikodym derivative is given by
\begin{equation*}
\frac{\dif \mathbb M}{\dif \mathbb N}|_{\mathcal F_T} = \frac{N_t M_T}{M_t N_T} =: L_T, \quad T\geq t,
\end{equation*}
$(L_T,\, T\geq t)$ is a martingale under $\mathbb N$. Girsanov Theorem gives a decomposition of local martingale in the $\mathbb N$ measure as a local martingale in the $\mathbb M$ measure plus a FV process:
\begin{theorem}[Girsanov]\label{theo:girsanov}
Let $\process X$ be a local martingale with respect to $\mathbb N$. Then the process
\begin{equation}
Y_t = X_t - \int_0^t \frac{\dif\ [X, L]_s }{L_s}
\end{equation}
is a local martingale with respect to $\mathbb M$.

If $[X,L]$ is $\mathbb M$--locally integrable, the process
\begin{equation}
Y_t = X_t - \int_0^t \frac{\dif\ \langle X, L\rangle_s }{L_s}
\end{equation}
is a local martingale with respect to $\mathbb M$.
\end{theorem}
We refer to \cite{JeanYor} for the proof of this theorem. An application of Girsanov theorem to Wiener processes is given in the following remark.

\begin{remark}
In case one knows that the Radon Nikodym derivative is given by
\begin{equation*}
\dif L_t = \sigma_LL_t\dif W_t^N
\end{equation*}
the decomposition of $\process{W^N}$ is easily obtained as
\begin{equation*}
\dif W^M = \dif W^N - \sigma_L \dif t
\end{equation*}
as $\dif\ \langle L, W^N\rangle_t = \sigma_L L_t \dif t$.
\end{remark}

The above Theorem will be applied to jump diffusion processes throughout this work. In particular, we will be considering Radon--Nikodym derivatives given by the Doleans--Dade exponential $\mathcal E(X)$ where
\begin{equation}\label{eqn:jdexpmart}
\dif X_t = \alpha \dif W_t + \beta \dif M_t,
\end{equation}
where $\alpha$ and $\beta$ will be either constant values or functions of $X_t$. More generally, 
the Doleans--Dade exponential $\mathcal E(X)$ is defined as the solution $Y$ of the following SDE
\begin{equation}\label{eqn:ddExp}
\dif Y_t  = Y_{t-}\dif X_t,\quad  Y_0 = 1.
\end{equation}
A proof that  a process like the one described in Eq \eqref{eqn:jdexpmart} is indeed a martingale, and not just a local martingale, is provided by Corollary 7 in \cite{ProtterShimbo}.
}

\section{Proof of Proposition \ref{prop:FX}}\label{app:proofFX}

\begin{proof}

The relation between $Z$ and $X$ is given by $X_t = \phi(Z_t)$ where $\phi(x) = 1/x$.	From Ito (see\iftoggle{submission}{, for example, \cite{JeanYor}}{ Appendix \ref{app:ito}})
\begin{align}
\dif X_t &= \dif \phi(Z_t)= \phi'(Z_{t-}) \dif Z_t + \frac 1 2   \phi''(Z_{t-})\dif\ [Z]^c_t + \sum_{s\leq t}\left( \left( \phi(Z_{s-} + \Delta Z_{s-}) - \phi(Z_{s-} ) \right) - \phi(Z_{s-}) {\Delta Z_{s-}}\right) \nonumber \\
&= \dif\left(\frac 1{Z_t}\right) = -\frac{\dif Z_t}{Z_t^2} + \frac{\dif\ [Z]^c_t}{Z_t^3} + \left( \frac{1}{Z_{t-} + \gamma^Z Z_{t-}} - \frac{1}{Z_{t-}}\right)\dif D_t + \frac{\Delta Z}{Z_t^2}\nonumber \\
&=- \bar \mu \frac{1}{Z_t}\dif t - \sigma^Z\frac{1}{Z_t}\dif  W_t^{(2)} - \gamma^Z \frac{1}{ Z_{t-}}\dif D_t + (\sigma^Z)^2\frac{1}{Z_t}\dif t + \frac{1}{Z_{t-}}\left( \frac{1}{1 + \gamma^Z } - 1 \right)\dif D_t + \frac{\gamma^Z }{Z_{t-}}\dif D_t \nonumber\\
&=- \bar \mu \frac{1}{Z_t}\dif t - \sigma^Z\frac{1}{Z_t}\dif  W_t^{(2)} + (\sigma^Z)^2\frac{1}{Z_t}\dif t + \frac{1}{Z_{t-}}\gamma^X \dif D_t \nonumber\\
&=- \bar \mu X_t\dif t - \sigma^ZX_t\dif  W_t^{(2)} + (\sigma^Z)^2X_t\dif t +X_{t-}\gamma^X\dif D_t
\end{align}
where we used $\gamma^X$ to denote the jumps of $\process X$, given by
\begin{equation}
\gamma^X = -\frac{\gamma^Z}{1+\gamma^Z}.
\end{equation}
We can now use Girsanov's Theorem in the form of Eq \ref{eqn:wienerTransform} for $\process{W^{(2)}}$ and Eq \ref{eqn:newLambda} for $\process D$ to decompose $\process X$ in a sum of local martingales in the new measure $\hat{\mathbb Q}$. As a result
\begin{align}
\dif X_t &=- \bar \mu X_t\dif t - \sigma^ZX_t(\dif  \hat W_t^{(2)} -\sigma^Z\dif t)  + (\sigma^Z)^2X_t\dif t +X_{t-}\gamma^X  (\dif \hat M_t  + (1+\gamma^Z)(1-D_t)\lambda_t \dif t)\nonumber \\
 &=- (\bar \mu -  \gamma^X(1+\gamma^Z)(1-D_t)\lambda_t \dif t)X_t\dif t - \sigma^ZX_t\dif  \hat W_t^{(2)}  +X_{t-}\gamma^X  \dif M_t  
\end{align}
Reminding that $\bar \mu$ is given by (see Eq \eqref{eqn:mu} and \eqref{eqn:compMu}) $ r -  \hat r - (1-D_t)\gamma^Z  \lambda_t$, the $\mathbb Q$-dynamics  of $\process Z$ can be written as
\begin{align}
\dif X_t &=( \hat r -  r + \gamma^Z(1-D) \lambda_t + \gamma^X(1+\gamma^Z)(1-D_t)\lambda_t ) X_t \dif t - \sigma^ZX_t\dif  \hat W_t^{(2)}  +X_{t-}\gamma^X  \dif \hat M_t  \nonumber \\
&=(\hat r -  r)X_t \dif t - \sigma^ZX_t\dif  \hat W_t^{(2)}  +X_{t-}\gamma^X  \dif \hat M_t.   \label{eqn:zJumpMart}
\end{align}

\end{proof}

\section{Proof of Proposition \ref{prop:gammaK}} \label{app:proof_gammaK}
\begin{proof}
Using Bayes' formula we can write 
\begin{equation*}
\E{t}{Z_T \1{\tau>T}} = \E{t}{Z_T |\ \1{\tau>T}} \E{t}{\1{\tau>T}} 
\end{equation*}
Under the dynamics given by \eqref{eqn:jumpFX}, the FX rate has only one jump at the default time of the reference entity,  therefore it is subject to no jumps  conditioned to the event $\1{\tau>T}$. 
This fact, together with the independence between the Brownian motions driving the FX and the hazard rate processes, allows to write:
\begin{equation*}
 \E{t}{Z_T |\ \1{\tau>T}} = Z_0 \E{t}{e^{\mu(T-t) - \gamma^Z \int_t^T \lambda_s\dif s}} =  Z_0 e^{\mu(T-t)} \E{t}{e^{- \gamma^Z \int_t^T \lambda_s\dif s}}
\end{equation*}
so that the survival probabilities are linked by
\begin{equation*}
\hat p_t(T) = \frac{\E{t}{Z_T \1{\tau>T}} }{Z_t}\frac{B(t,T)}{\hat B(t,T)} = \E{t}{e^{- \gamma^Z \int_t^T \lambda_s\dif s}} p_t(T).
\end{equation*}
The above can be written in terms of default probabilities,
\begin{equation*}
1 - \hat p_t(T) = 1 - \E{t}{e^{- \gamma^Z \int_t^T \lambda_s\dif s}} p_t(T)= 1 - p_t(T) + \left(1- \E{t}{e^{- \gamma^Z \int_t^T \lambda_s\dif s}} \right)p_t(T),
\end{equation*}
so that the ratio of default probabilities can be expressed as
\begin{align}
\frac{1 - \hat p_t(T) }{1 - p_t(T) } &= 1 + \left(1- \E{t}{e^{- \gamma^Z \int_t^T \lambda_s\dif s}} \right)\frac{p_t(T)}{1 - p_t(T) }\nonumber \label{eqn:appC}\\
&= 1 + \left(1- \E{t}{e^{- \gamma^Z \int_t^T \lambda_s\dif s}} \right)\frac{\E{t}{e^{- \int_t^T \lambda_s\dif s}}}{1 - \E{t}{e^{- \int_t^T \lambda_s\dif s}} } .
\end{align}

Given that our aim is to find an approximation for small maturities, it is convenient to note that
\begin{equation}
e^{-\int_0^T\lambda_s \dif s } = 1 - \lambda_0 T + O(T^2) \quad\mbox{as }T\rightarrow 0
\end{equation}
so that we can write the right hand side of Eq \eqref{eqn:appC} at $t = 0$  as
\begin{align*}
\textrm{rhs Eq \eqref{eqn:appC}} &= 1+\left(1- \E{0}{1 - \gamma^Z  \lambda_0 T + O(T^2)} \right)\frac{\E{0}{1- \lambda_0 T+ O(T^2)} }{ \E{0}{ \lambda_0 T+ O(T^2)}  }\\
&= 1+\gamma^Z({ \lambda_0T + \mathbb E_0[O(T^2)]} )\frac{{1-  \lambda_0 T+ \mathbb E_0[O(T^2)]} }{ {  \lambda_0 T+ \mathbb E_0[O(T^2)]}  }\\
&= 1+\gamma^Z (1-  \lambda_0 T+ \mathbb E_0[O(T^2)])  \frac{\lambda_0T }{ {  \lambda_0 T+ \mathbb E_0[O(T^2)]}  } + \gamma^Z{ \mathbb E_0[O(T^2)]} \frac{{1-  \lambda_0 T+ \mathbb E_0[O(T^2)} }{ {  \lambda_0 T+ \mathbb E_0[O(T^2)]}  }\\
&= 1+\gamma^Z (1-  \lambda_0 T+ \mathbb E_0[O(T^2)])  \frac{1 }{  1+\frac{ \mathbb E_0[O(T^2)]}{\lambda_0T}  } + \gamma^Z{ \mathbb E_0[O(T^2)]} \frac{{1-  \lambda_0 T+ \mathbb E_0[O(T^2)]} }{ {  \lambda_0 T+ \mathbb E_0[O(T^2)]}  }
\end{align*}
from which we have
\begin{equation}
\frac{1 - \hat p_0(T) }{1 - p_0(T) } \rightarrow 1 + \gamma^Z,\quad\mbox{as }T\rightarrow 0.
\end{equation}
\end{proof}

\end{appendices}

\section*{{Disclaimers}}

The opinions and views are uniquely those of the authors and do not necessarily represent those of their employers.

\bibliographystyle{plainnat}
\bibliography{MYBIB}

\end{document}